\documentclass[11pt]{article}
\pdfoutput=1

\usepackage{amsmath, amssymb,amsthm}
\usepackage{comment}
\newtheorem{theorem}{Theorem}

\newtheorem{prop}[theorem]{Proposition}

\theoremstyle{definition}

\numberwithin{equation}{section}

\usepackage{euscript}
\usepackage{amsfonts}
\usepackage{amsbsy}
\usepackage{epsfig}
\usepackage{amsthm}
\usepackage{amscd}
\usepackage{amstext}
\usepackage{verbatim}
\usepackage{cancel}
\usepackage{capt-of}
\usepackage{empheq}

\usepackage[nosort]{cite}
\usepackage{bm}
\usepackage{authblk}
\usepackage{esint}
\usepackage[T1]{fontenc}
\usepackage{mathdots}
\usepackage[centertableaux]{ytableau}

\usepackage[utf8]{inputenc}
\usepackage{graphicx}
\usepackage{physics}
\usepackage{mathtools}
\usepackage{bbold}

\usepackage{setspace}
\usepackage{colortbl}
\usepackage{xcolor}

\usepackage{amsmath}
\makeatletter
\renewcommand*\env@matrix[1][\arraystretch]{%
  \edef\arraystretch{#1}%
  \hskip -\arraycolsep
  \let\@ifnextchar\new@ifnextchar
  \array{*\c@MaxMatrixCols c}}
\makeatother
\usepackage{amssymb}
\usepackage{mathrsfs}
\usepackage{booktabs}
\usepackage{bbm}
\usepackage{float}

\usepackage{pgfplots}
\pgfplotsset{compat=1.15}
\usepackage{tikz}
\usetikzlibrary{external}
\usetikzlibrary{arrows}

\usepackage{subcaption}
\usepackage{graphicx}
\usepackage{mathtools}
\usepackage[many]{tcolorbox}
\tcbset{shield externalize}
\usepackage{xcolor, colortbl}
\usepackage[a4paper]{geometry}
\usepackage[hidelinks,linktocpage]{hyperref}

\hypersetup{
    colorlinks=true,
    linkcolor=blue,
    citecolor=blue,
    }

\usepackage{mathrsfs}

\def\ben{\begin{equation}}
\def\een{\end{equation}}

\let\a=\alpha \let\b=\beta \let\g=\gamma  
   \let\k=\kappa
\let\l=\lambda  \let\n=\nu

\let\X=\Xi   
\let\C=\Chi

\def\be{\begin{equation}}
\def\ee{\end{equation}}
\def\beq{\begin{equation}}
\def\eeq{\end{equation}}
\def\ba{\begin{array}}
\def\ea{\end{array}}

\def\dalemb#1#2{{\vbox{\hrule height .#2pt
       \hbox{\vrule width.#2pt height#1pt \kern#1pt
               \vrule width.#2pt}
       \hrule height.#2pt}}}

\newcommand{\bea}{\begin{eqnarray}}
\newcommand{\eea}{\end{eqnarray}}

\makeatletter
\newcommand*\bigcdot{\mathpalette\bigcdot@{.5}}
\newcommand*\bigcdot@[2]{\mathbin{\vcenter{\hbox{\scalebox{#2}{$\m@th#1\bullet$}}}}}
\makeatother

\def\R{{{\mathbb R}}}
\def\C{{{\mathbb C}}}

\def\RP{{{\mathbb R}{\mathbb P}}}

\title{Charged and rotating near-horizon geometries in five dimensions}
\author{Alex Colling\thanks{Email: \texttt{aec200@cam.ac.uk}} \ and Jun~Liu\thanks{Email: \texttt{jl2296@cam.ac.uk}}}
\affil{\it Department of Applied Mathematics and Theoretical Physics, \\
\it University of Cambridge, Cambridge CB3 0WA, UK
}
\date{}

\begin{document}

\maketitle

\begin{abstract}

We present new charged and rotating near-horizon geometries in five-dimensional Einstein-Maxwell theory in closed analytic form. The solutions can be parametrised by the charge and two independent angular momenta. We also generalise these near-horizon geometries to theories with an additional Chern-Simons term in the action multiplied by an arbitrary coupling constant. The new solutions have the same entropy relations as expected for charged  versions of extremal Myers-Perry black holes and for rotating versions of extremal Reissner-Nordstr\"om-Tangherlini black holes, but they do not reduce to the Myers-Perry horizon in the vacuum limit. The horizon cross-sections are spherical and carry a Sasakian structure. We exploit this structure to prove a characterisation of our solutions: without any symmetry assumptions, they are the most general rotating extremal horizons for which the co-rotating electric field is a (non-zero) constant. We further extend this construction to higher dimensions, where we show that any Sasaki-Einstein manifold generates a two-parameter family of charged and rotating horizons. 

\end{abstract}
\newpage

\tableofcontents
\section{Introduction}
A generalisation of the Kerr-Newman family of black holes to five-dimensional Einstein-Maxwell theory is not analytically known. This family of solutions should depend on four parameters: a mass, two independent angular momenta, and a charge. However, only special limits with either zero charge (Myers-Perry \cite{Myers:1986un}) or zero angular momenta (Reissner-Nordstr\"om-Tangherlini \cite{Tangherlini:1963bw}) are known in closed form. Charged and rotating solutions have only been constructed numerically \cite{Kunz:2005nm, Horowitz:2024kcx} or perturbatively \cite{Aliev:2006yk}, see also \cite{Deshpande:2024vbn}. If one allows for an additional Chern-Simons term and fixes the Chern-Simons coupling to align with the bosonic part of five-dimensional minimal supergravity, a four-parameter family of solutions interpolating between these two limits is known analytically (Chong-Cvetic-L\"u-Pope \cite{Chong:2005hr}).  

Restricting to extremal black holes, instead of the full black hole spacetime one can attempt to construct just the near-horizon geometry describing the intrinsic geometry of the (extremal) event horizon. This is a more tractable problem due to the presence of additional symmetries: the near-horizon geometry is invariant under a $SO(2,1)$ action preserving the Maxwell field \cite{Kunduri:2007vf, Dunajski:2023xrd, Colling:2024txz, Colling:2025dub}. The five-dimensional extremal Kerr-Newman horizon should depend on three parameters (a charge and two angular momenta), with the mass determined by the extremality condition. Cross-sections of the horizon are topologically three-spheres\footnote{Various charged and rotating extremal black holes and near-horizon geometries with horizon topology $S^2 \times S^1$ have been constructed, see e.g. \cite{Emparan, Kunduri:2011zr}. We will focus on $S^3$ topology for most of this work.} $S^3$. Perhaps somewhat surprisingly, just like for the full spacetime only special cases of such a three-parameter family of  near-horizon geometries are known analytically. These are the following \cite{Kunduri:2013gce}:
\begin{itemize}
    \item \textbf{Vacuum}: for vanishing Maxwell field, the solution should reduce to the two-parameter family of extremal Myers-Perry horizons. In fact, the Myers-Perry horizons are part of a three-parameter family of known vacuum near-horizon geometries with $U(1)^2$-symmetric spherical cross sections \cite{KLvac, Hollands:2010bf} that also contain the horizons of the Rasheed black holes (with Kaluza-Klein asymptotics)\cite{KK}. 
    \item \textbf{Static}: there is a known two-parameter family of static horizons arising from static extremal black holes in a background electric field \cite{KL}. This family contains the extremal Reissner-Nordstr\"om-Tangherlini horizon as a special case and is the most general static near-horizon geometry with spherical $U(1)^2$-symmetric cross-sections.
    
    \item \textbf{Homogeneous}: when the two angular momenta are equal, the horizon cross-sections are homogeneous with isometry group containing $U(2)$. There are two such two-parameter families of near-horizon geometries \cite{Kunduri:2013gce, homog2}. One family has a vacuum limit (homogeneous Myers-Perry), and  the other family has a static limit (Reissner-Nordstr\"om-Tangherlini). The two families meet at a common one-parameter subfamily.
\end{itemize}
In the presence of a Chern-Simons term with the particular coupling corresponding to five-dimensional minimal supergravity, many more spherical near-horizon geometries are known. This includes the extremal charged Myers-Perry horizons in \cite{Chong:2005hr} and the horizons of charged Kaluza-Klein black holes \cite{six}. However, within five-dimensional pure Einstein-Maxwell theory, to the best of our knowledge, there is no analytically known charged and rotating spherical near-horizon geometry with two independent angular momenta.

In a recent work \cite{Horowitz:2024kcx}, Horowitz and Santos numerically constructed multiple three-parameter families of solutions with independent angular momenta and charge for any value of the Chern-Simons coupling. These solutions admit a cohomogeneity-one $U(1)^2$ isometric action on the cross-sections and generalise the homogeneous horizons in \cite{Kunduri:2013gce}. For vanishing Chern-Simons coupling, their solutions interpolate in two patches between the static and vacuum limits above, and therefore are expected to represent the near-horizon geometry of the extremal Kerr-Newman family in five dimensions. Remarkably, the numerical solutions satisfy two simple closed form entropy relations, hinting at the existence of closed analytic expressions for the near horizon geometries. We present these entropy relations later in equations \eqref{eq:entropy1} and \eqref{eq:Ent1}.

\subsection{Summary of results}

In this work we present new charged and rotating near-horizon geometries in five-dimensional Einstein-Maxwell-Chern-Simons theory. For zero Chern-Simons coupling, we construct two three-parameter families of solutions with $U(1)^2$ symmetry interpolating between the static horizons in \cite{KL} and a two-parameter family of vacuum horizons. The known homogeneous solutions are recovered by setting the angular momenta equal. 

To present the solutions, we use  coordinates $y \in [0,1]$, $\phi \sim \phi + 2\pi$, and $\psi \sim \psi + 2\pi$ on the spherical cross-sections. Details on conventions are explained below in \S \ref{sec:known}. Family 1 (details in \S \ref{sec:fam1}) has a non-rotating limit ($q^2 = c_1 c_2$ below) coinciding with the static squashed horizons in \cite{KL}. The metric and Maxwell field are
\begin{subequations}\label{eq:gF}
\begin{multline}
g = \frac{\Gamma}{4} \left(-\rho^2 \dd t^2 + \frac{\dd \rho^2}{\rho^2} \right)
+ \frac{\Gamma}{4y(1-y)} \text{d}y^2
+ \frac{c_1^2 \left[q^2 (1-y)+ c_1^2 y \right]}{\Gamma} (1-y) E_\phi^2\\
- \frac{2 c_1c_2 (c_1c_2 - q^2)}{\Gamma}y(1-y)E_\phi E_\psi
+ \frac{c_2^2 \left[c_2^2 (1-y) + q^2 y \right]}{\Gamma} y E_\psi^2,
\label{eq:g}
\end{multline}

\begin{equation}
F = \frac{\sqrt{3}\,q}{4}\text{d}t \wedge \text{d}\rho
+ \frac{\sqrt{3} \sqrt{c_1 c_2 - q^2}}{2 \Gamma}
\left( c_1^2  \text{d}y \wedge E_\phi
- c_2^2 \text{d}y \wedge E_\psi \right).
\label{eq:F}
\end{equation}
with parameters $c_1, c_2>0$, $c_1 c_2 \geq q^2$, and 
\be
    \Gamma = [c_1 y + c_2 (1-y)]^2, \quad
    E_\phi = \dd \phi + \frac{c_2 \sqrt{c_1 c_2 - q^2}}{2c_1 q} \rho \dd t, \quad
     E_\psi = \dd \psi + \frac{c_1\sqrt{c_1 c_2 - q^2}}{2c_2 q} \rho \dd t.
\ee
\end{subequations}
Family 2 (details in \S \ref{sec:fam2}) has a zero charge limit  ($q^2 = 2 c_1 c_2$ below) intersecting the three-parameter family of vacuum solutions at a two-parameter slice. However, this vacuum slice is not the extremal Myers-Perry slice, but a different two-parameter slice containing Kaluza-Klein black hole horizons. Nonetheless, perhaps surprisingly, these solutions also satisfy the simple entropy relations anticipated in \cite{Horowitz:2024kcx} for geometries connected to the extremal Myers-Perry horizon. The metric takes the same form as \eqref{eq:g}, only $E_\phi$, $E_\psi$ and $F$ are different:
\begin{subequations}\label{eq:gF2}
\begin{equation}
F = \frac{\sqrt{2c_1c_2-q^2}}{4}\text{d}t \wedge \text{d}\rho
+ \frac{\sqrt{2c_1c_2-q^2}}{2 \Gamma}
\left( c_1^2  \text{d}y \wedge E_\phi
- c_2^2 \text{d}y \wedge E_\psi \right).
\label{eq:F2}
\end{equation}
with parameters $c_1, c_2>0$, $2 c_1c_2 \geq q^2$, and 
\be
    E_\phi = \dd \phi + \frac{c_2}{2c_1} \rho \dd t, \quad
     E_\psi = \dd \psi + \frac{c_1}{2c_2} \rho \dd t.
\ee
\end{subequations}

Additionally, we present three-parameter families of solutions for any Chern-Simons coupling. Their entropy relations are not as simple, but also qualitatively and perturbatively (to leading order in angular momentum or charge) agree with those presented in \cite{Horowitz:2024kcx}. For the supergravity coupling we again find two three-parameter families, one of which has both a vacuum and a static limit. Just like for vanishing Chern-Simons coupling, this family has the same entropy relation as the charged Myers-Perry (CCLP) horizons, even though these are different solutions. The explicit expressions for the solutions are deferred to \S\ref{sec:genlam}.

More generally, we show that the equations of motion for $U(1)^2$ symmetric extremal horizons can be reduced to three coupled second order ODEs \eqref{bd:es} together with a first order constraint \eqref{eq:cons}. We deduce that any family of solutions with spherical horizon cross-sections and fixed Chern-Simons coupling can depend on at most five continuous parameters. The three-parameter families above are obtained by explicitly solving the system of ODEs under the assumption that the co-rotating electric field\footnote{The co-rotating electric field is defined as the function multiplying the AdS$_2$ volume form $\frac 14\text{d}t \wedge \text{d}\rho$ in the expression for the near-horizon Maxwell field.} is constant. 

It turns out that there is an underlying Sasakian structure on the spherical cross-sections of our solutions, which allows us to prove a classification result characterising them. Any rotating horizon admits at least one Killing vector field tangent to (compact) cross-sections as a consequence of the rigidity theorem proven in \cite{Dunajski:2023xrd, Colling:2025dub}. By performing a conformal rescaling on the orbit space of this Killing vector, we construct a Sasakian structure on the cross-section for solutions with constant non-zero co-rotating electric field. The Sasakian structure can be used to establish the existence of a second commuting Killing vector (Proposition \ref{prop2kvf}), resulting in the following classification of these solutions.

\begin{theorem}\label{result:all}
    Any smooth five-dimensional rotating near-horizon geometry with constant non-zero co-rotating electric field and compact horizon cross-sections is given by one of the solutions presented in \textup{\S\ref{sec:EMsol}}. In particular, within pure Einstein-Maxwell theory any such solution is (up to a quotient) contained in the family \eqref{eq:gF} or in \eqref{eq:gF2}.
\end{theorem}

We emphasize that this result does not a priori assume the existence of any symmetries on the horizon cross-sections. Conversely, in odd dimensions higher than five it is possible to reverse the procedure leading to a Sasakian structure whenever the conformal transformation involved is constant. In this way one can construct a two-parameter family of $(2n+3)$-dimensional charged and rotating near-horizon geometries starting from a $(2n+1)$-dimensional Sasaki-Einstein manifold. This may be viewed as a generalisation of the construction of vacuum horizons in \cite{Kunduri:2012uq}. As an example, we explicitly present the solutions corresponding to the standard Sasaki-Einstein structure on a round $S^{2n+1}$.

\paragraph{Outline} In \S \ref{sec:known} we introduce five-dimensional Einstein-Maxwell-Chern-Simons theory and write down the equations of motion governing near-horizon geometries in this theory, the \textit{near-horizon equations}, using the formalism of \cite{Dunajski:2023xrd, Colling:2025dub}. In \S \ref{sec:solve}, we study the near-horizon equations, first without symmetry assumptions and then for geometries invariant under a $U(1) \times U(1)$ action. We prove an upper bound on the number of solutions and construct our new solutions with a constant co-rotating electric field Ansatz. Readers interested only in the solutions but not the details of solving the equations may jump directly to \S \ref{sec:EMsol} without difficulty. In \S \ref{sec:EMsol}, we present the new solutions, their conserved charges, and thermodynamic relations. Special values of the Chern-Simons coupling are treated separately. In \S \ref{sec:sas}, we consider the Sasakian structure underlying these new solutions and complete the proof of Theorem \ref{result:all}. We then present the construction of higher-dimensional generalisations of these horizons starting from a Sasaki-Einstein manifold. Appendix \ref{Aknown} contains an exposition of the previously known near-horizon geometries relevant to this work. 

\paragraph{Note Added.} While we were in the final stages of preparing this manuscript, we learned that the family of solutions we call family 1 (equation \eqref{eq:gF}) in pure Einstein-Maxwell theory has been independently found by Cameron Gibson.

\section{Five-dimensional Einstein-Maxwell-Chern-Simons theory}\label{sec:known}
We consider the following action
\be\label{eq:action}
S = \int d^5 x \sqrt{-g} \left(R - F_{\mu \nu} F^{\mu \nu} \right) - \frac{8 \l}{3 \sqrt{3}} \int F \wedge F \wedge A, 
\ee
where $R$ is the Ricci scalar of the five-dimensional spacetime metric $g$ and $F = \text{d}A$ is the Maxwell field strength. The Chern-Simons coupling $\l$ is normalised such that the action becomes the bosonic part of $5D$ minimal supergravity at $\l = 1$. We take $\l$ real and non-negative throughout this work, since the theory with negative $\l$ can be mapped to positive $\l$ by simultaneously negating the Maxwell field. The equations of motion read
\begin{subequations}\label{bd:fulleqs}
\begin{align}
&R_{\mu\nu} = 2 \left( F_{\mu\rho}F_{\nu}^{\:\:\rho}-\frac 16 F_{\rho\sigma}F^{\rho\sigma}g_{\mu\nu} \right), \label{eq:eqgrav}\\[5pt]
&\text{d} \star F + \frac{2 \l}{\sqrt{3}} F \wedge F =0 . \label{eq:eqmat}
\end{align}
\end{subequations}

For stationary asymptotically flat spacetimes with $U(1)^2$ axial symmetry, the mass, charge, and angular momenta can be computed via Komar integrals \cite{Komar} as
\begin{equation}
    M = \frac{3}{32 \pi} \int_{S^3_\infty} \star \text{d} k, \quad J_i = - \frac{1}{16 \pi} \int_{S^3_\infty} \star \text{d} m_i, \quad Q = \frac{1}{8 \pi} \int_{S_\infty^3} \left( \star F + \frac{2 \l}{\sqrt{3}} A \wedge F \right), 
\end{equation}
where $k$ is the stationary Killing vector, and $m_{1,2}$ are two commuting axial Killing vectors. There exist various definitions of electric charge when the Chern-Simons coupling is non-zero. The charge $Q$ here is the Page charge \cite{Page}, which is conserved. This charge is well-defined provided we choose the gauge field $A$ to be globally defined on the three-sphere $S^3_\infty$ at spatial infinity.

If the spacetime includes a black hole region with a Killing event horizon $\mathcal{H}$, the generator of the horizon can be decomposed as
\be
\xi = k + \sum_i \Omega_i m_i. 
\ee
The $\Omega_i$ are interpreted as the angular velocities of the black hole as measured by a stationary observer. We assume that topologically $\mathcal{H} = \mathbb{R} \times H$ with $H$ a compact three-dimensional cross-section transversal to the integral curves of $\xi$. It was shown in \cite{Galloway:2005mf} (see also \cite{Lucietti:2012sa} for the extremal case) that the horizon cross-section is diffeomorphic to a three-sphere $S^3$, a ring $S^2 \times S^1$, a lens space $L(p,q)$ or a connected sum of these\footnote{Assuming $U(1)^2$ symmetry the connected sums can be ruled out \cite{topsum}.}. We focus on $S^3$ topology below. Stationary and asymptotically flat spherical black hole solutions to~\eqref{eq:action} with $U(1)^2$ axial symmetry satisfy a Smarr formula as
\be\label{eq:smarr}
\frac{2}{3} M =   \sum_i \Omega_i J_i + \frac{\k \mathcal{A}}{8\pi} + \frac{4}{3} \Phi_{\mathcal{H}} Q.
\ee
Here $\k$ is the surface gravity, $\mathcal{A}$ is the area of the horizon cross-section, and $\Phi_{\mathcal{H}}$ denotes the co-rotating electric potential $\Phi = \iota_\xi A$ evaluated at the horizon $\mathcal{H}$. This requires fixing a gauge such that $\mathcal{L}_\xi A = 0$ and $\Phi_{S^3_\infty} = 0$. We have further assumed that $H_2(\Sigma) = 0$ for a spacelike hypersurface $\Sigma$ interpolating between the asymptotic $S^3_\infty$ and a horizon cross-section $H$. The general formula including ring topology for the horizon and possible solitons outside the horizon was derived in \cite{Kunduri:2013vka}. 

Since we are interested in the geometry of the horizon, it is helpful to use conservation laws to relate Komar integrals for computing the conserved quantities at infinity to integrals over the horizon cross-section $H$. The Page charge is conserved by itself, so it can be calculated on any compact spacelike three-surface cobordant to $S_\infty^3$, including $H$, 
\be \label{eq:defcharge}
Q = \frac{1}{8\pi} \int_H \left( \star F + \frac{2 \l}{\sqrt{3}} F \wedge A \right). 
\ee
The expressions for angular momenta require some algebra, and in the end they are given by \cite{Horowitz:2024kcx,Kunduri:2013gce,Hanaki:2007mb}
\be \label{eq:defangmomt}
J_{i} = -\frac{1}{16\pi} \int_H \left(\star \text{d} m_i + 4 (m_i \cdot A) \star F + \frac{16 \l}{3\sqrt{3}}(m_i \cdot A) A \wedge F\right),
\ee
where $A$ is in a global gauge satisfying $\mathcal{L}_{m_i} A = 0$ (such a gauge may be obtained by averaging over the orbits of $m_i$ \cite{Hanaki:2007mb}), and the $\cdot$ denotes contraction of indices. The Komar mass $M$ cannot be computed from data on the horizon alone because it requires information from infinity about the stationary Killing vector.

\subsection{Extremal horizons} \label{ehorsec}

In this work we consider the case where the horizon $\mathcal{H}$ is extremal, i.e. the surface gravity vanishes. The spacetime is assumed to be smooth. Although above we assumed $\mathcal{H}$ is the event horizon of a black hole, our analysis does not rely on global properties of the spacetime and hence applies to any extremal Killing horizon in the theory (\ref{eq:action}). It is well known that in this situation one can define an associated spacetime, the near-horizon geometry \cite{Reall:2002bh,Kunduri:2013gce} , which itself solves the equations \eqref{bd:fulleqs} and may be thought of as a scaling limit of the original spacetime zooming into the horizon. In Gaussian null coordinates $(v,r,x^i)$, it reads
\begin{subequations}\label{bd:nh}
\begin{align}
    g &= 2\text{d}v\left(\text{d}r + rX_i(x)\text{d}x^i + \tfrac 12G(x)r^2\text{d}v^2\right) + \gamma_{ij}(x)\text{d}x^i\text{d}x^j, \label{gnh}\\
    F &= -\text{d}\left(\psi(x)r\text{d}v\right) + \tfrac 12B_{ij}(x)\text{d}x^i \wedge \text{d}x^j, \label{fnh}
\end{align}  
\end{subequations}
which is globally defined on $\R^2_{v,r} \times H$. The horizon is at $r = 0$ with generator $\xi = \partial_v$ and the $x^i$ are local coordinates on the horizon cross-section $H \cong S^3$. The near-horizon geometry is determined by data $(\gamma, X,\psi,B)$ on $H$, where $\gamma$ is the induced (Riemannian) metric, $X = X_i(x)\text{d}x^i$ is a 1-form, $\psi$ a function and $B = \frac 12B_{ij}(x)\text{d}x^i\wedge \text{d}x^j$ a closed 2-form. The equations \eqref{bd:fulleqs} imply the following near-horizon equations on $H$,
\begin{subequations}\label{bd:he}
\begin{align}
    R_{ab} = \frac 12 X_aX_b - \nabla_{(a}X_{b)} + 2B_{ac}B_b^{\:\:c} + \frac 13\gamma_{ab}(2\psi^2 - B_{cd}B^{cd}) &, \label{he1}\\
    (\nabla^a - X^a)B_{ab} = -(\nabla_b - X_b)\psi  + \frac{4\lambda}{\sqrt{3}}\psi(\star B)_b&. \label{he2}
\end{align}
\end{subequations}
Here $R_{ab}$ is the Ricci tensor of the Levi-Civita connection $\nabla$ of the metric $\gamma$ and $\gamma$ is used to raise indices and define the Hodge star operator $\star$. Conversely, given a solution to \eqref{bd:he} we can reconstruct the function $G$ and hence the full near-horizon geometry using 
\begin{equation}
    G = \frac 12\vert X \vert^2 - \frac 12\nabla_a X^a - \frac 43\psi^2 - \frac 13\vert B \vert^2,
\end{equation}
where $\vert \cdot \vert$ denotes the $\gamma$-norm. 

For any solution to \eqref{bd:he} there is a smooth function $\Gamma > 0$ (unique up to constant rescaling) on $H$ such that 
\begin{equation} \label{Kpref}
    K^a = \Gamma X^a + \nabla^a\Gamma
\end{equation}
either vanishes identically or is a Killing vector of $\gamma$ preserving $(X,\psi,B)$. This was first shown in the vacuum case by Dunajski and Lucietti \cite{Dunajski:2023xrd} and later extended to four-dimensional Einstein-Maxwell theory \cite{Colling:2024txz} and to general matter theories including (\ref{eq:action}) \cite{Colling:2025dub}. In terms of $K$ and $\Gamma$, the near-horizon equations \eqref{bd:he} reduce to
\begin{subequations}\label{bd:nhe}
\begin{align}
    &R_{ab} = \frac{1}{2\Gamma^2}K_aK_b - \frac{(\nabla_a\Gamma)(\nabla_b\Gamma)}{2\Gamma^2} +\frac{1}{\Gamma}\nabla_a\nabla_b\Gamma+2B_{ac}B_b^{\:\:c} + \frac 13 \g_{ab}(2\psi^2-B_{cd}B^{cd}),\label{nhe1} \\
    &\nabla^a(\Gamma B_{ab}) = K_b\psi + \frac{4\l}{\sqrt{3}} \Gamma \psi (\star B)_b, \label{nhe2}\\
    &K^aB_{ab} = \nabla_b(\psi\Gamma). \label{nhe3}
\end{align}
\end{subequations}
This is the formulation of the equations of motion we will study in this work. They imply that the near-horizon geometry \eqref{bd:nh} can be written in terms of the coordinate $\rho = r\Gamma^{-1}$ as 
\begin{subequations}\label{bd:nh2}
\begin{align}
    g &= \Gamma[2\text{d}v\text{d}\rho - \alpha \rho^2\text{d}v^2] + \g_{ij}(\text{d}x^i + K^i\rho\text{d}v)(\text{d}x^j + K^j\rho\text{d}v), \label{gnh2} \\
    F &= \Gamma\psi \:\text{d}v \wedge \text{d}\rho +\tfrac 12 B_{ij}(\text{d}x^i + K^i\rho\text{d}v)\wedge(\text{d}x^j + K^j\rho\text{d}v). \label{fnh2}
\end{align}
\end{subequations}
Here $\alpha$ is defined by
\begin{equation} \label{alpha}
    \alpha = \frac{\vert K \vert^2}{2\Gamma} - \frac 12\Delta \Gamma + \frac 43 \Gamma \psi^2 + \frac 13\Gamma \vert B \vert^2 = \frac{\vert K \vert^2}{\Gamma} - \Gamma G,
\end{equation}
which can be proven to be constant by the equations of motion. An integration over $H$ shows that $\alpha > 0$ unless $K,B,\psi$ all vanish, so that the metric in the square brackets in \eqref{gnh2} is locally isometric\footnote{The coordinate $v$ is related to $t$ in \eqref{eq:g} (where we fixed $\alpha = 4$) via $\alpha v = t - \rho^{-1}$.} to AdS$_2$. The (maximally extended) near-horizon geometry is invariant under a $SO(2,1)$ action corresponding to the isometries of this AdS$_2$ factor \cite{Kunduri:2007vf}. Note that we did not assume any symmetries in the original spacetime other than the horizon generator $\xi$.

We call a solution rotating if $K$ is non-zero (see \cite{Colling:2025dub}). A solution is static if $\xi = \partial_v$ is hypersurface-orthogonal in the near-horizon geometry. 
We refer to the function $\chi = \Gamma\psi$ multiplying the volume form $\text{d}v\wedge\text{d}\rho$ of the AdS$_2$ factor in the expression (\ref{fnh2}) for $F$ as the co-rotating electric field and to the induced Maxwell field $B$ on $H$ as the magnetic field.

\subsection{Extremal Smarr formula}\label{sec:extsmarr}
Within the near-horizon geometry, one could take \eqref{eq:defcharge} and \eqref{eq:defangmomt} as definitions of charge and angular momenta in the presence of rotational Killing fields without worrying about their connections to Komar integrals at asymptotic infinity\footnote{If the near-horizon geometry does not arise from an asymptotically flat extremal black hole spacetime, the conserved charges computed on the horizon as in (\ref{eq:defcharge}, \ref{eq:defangmomt}) may be different from the asymptotic Komar charges. This would be the case, for example, for horizons of asymptotically Kaluza-Klein black holes.}. See \cite{Hajian:2013lna,Donnay:2016ejv} for alternative derivations of these formulae from near-horizon symmetries. In \cite{Hajian:2013lna} it is shown that these quantities satisfy the following relation within the near-horizon geometry 
\be\label{eq:Smarr}
\sum_i \omega_i J_i + \frac{4}{3} \, \mu \, Q  + \frac{\alpha}{8\pi} \mathcal{A} = 0,  
\ee
where $\mathcal{A}$ is the area of the horizon cross-section. Here $\omega_i$ are constants such that the preferred Killing vector (\ref{Kpref}) in the near-horizon geometry is $K = \sum_i\omega_i m_i$. The constant $\mu$ is defined as
\be \label{mu}
\mu \equiv  \Gamma  \psi + K \cdot C, 
\ee 
where $C$ is the induced gauge potential on $H$ satisfying $B = \text{d}C$. We require it to be in a gauge such that $\mathcal{L}_K C = 0$ and it is globally well-defined on $H$. The expression for $\mu$ is constant due to the matter equation  \eqref{nhe3}. Note that $\omega_i, \mu$ and $\alpha$ are all defined up to simultaneous rescaling by a positive constant, which can be traced back to the scaling freedom in the definition of $\Gamma$.

We will present our own self-contained derivation of (\ref{eq:Smarr}) using the formalism introduced in \S\ref{ehorsec}. It involves writing the charge and angular momenta in near-horizon variables. Inserting \eqref{bd:nh} into the expressions (\ref{eq:defcharge}, \ref{eq:defangmomt}) yields
\begin{gather}
    Q = - \frac{1}{8\pi} \int \text{d}^3x \sqrt{\gamma} \, \left(\psi - \frac{2\l}{\sqrt{3}} C \cdot (\star B)\right), \label{Qhor} \\[5pt]
    J_i = - \frac{1}{16\pi} \int \text{d}^3 x \sqrt{\gamma} \left( X_i - 4 C_i \, \psi + \frac{16 \l}{3 \sqrt{3}} C_i \, C \cdot (\star B)\right). \label{Jhor}
\end{gather}
where $C_i = m_i \cdot C$ and similarly for $X_i$. Since $\mathcal{L}_{m_i}\Gamma = 0$ we have $X_i = \Gamma^{-1}K_i$ from (\ref{Kpref}), so that
\begin{gather}
    \mu Q = -\frac{1}{8 \pi} \int \text{d}^3 x \sqrt{\gamma} \left(\Gamma \psi^2 + (K \cdot C) \psi - \frac{2 \l}{\sqrt{3}} (\Gamma \psi  + K \cdot C)[C \cdot (\star B)]\right), \\[5pt]
    \sum_{i} \omega_i J_i = - \frac{1}{16 \pi} \int \text{d}^3 x \sqrt{\gamma} \left( \frac{|K|^2}{\Gamma} - 4 ( K \cdot C) \psi + \frac{16\l}{3 \sqrt{3}} (K \cdot C)[C \cdot (\star B)] \right).
\end{gather}
Next we use the matter equation of motion \eqref{nhe2} to substitute
\be
( K \cdot C) \psi \:\Hat{=} - \frac{1}{2} \Gamma |B|^2 - \frac{4 \l}{\sqrt{3}} \Gamma \psi [ C \cdot (\star B)], 
\ee
which holds up to a divergence $\nabla^a(\Gamma A^bB_{ab})$ that drops out after integration. Then the first two terms in \eqref{eq:Smarr} are
\be
\begin{aligned}
\sum_i \omega_i J_i + \frac{4}{3} \mu Q &= - \frac{1}{8\pi} \int \text{d}^3 x \sqrt{\gamma} \left[\frac{|K|^2}{2\Gamma} + \frac{4}{3} \Gamma \psi^2 + \frac{1}{3} \Gamma |B|^2 \right]\\
& = -\frac{\alpha}{8\pi} \mathcal{A}. 
\end{aligned}
\ee
In going to the second line, we used \eqref{alpha} and the fact that the horizon is compact to ignore the integral of $\Delta\Gamma$. Terms proportional to $\l$ cancel out and we arrive at \eqref{eq:Smarr}.

If the near-horizon geometry arises from an extremal black hole spacetime which may be viewed as a limiting member of a subextremal family of solutions, one can think of \eqref{eq:Smarr} as a ``zero temperature limit'' of the Smarr formula \eqref{eq:smarr} \cite{Hajian:2013lna}. To see this, let $\Omega_i, \Phi_{\mathcal{H}}$ and $\kappa$ denote the angular velocities, co-rotating electric potential and surface gravity of the subextremal solutions, with extremal values $\Omega_i^{ext.}, \Phi_{\mathcal{H}}^{ext.}$, and $\kappa^{ext.} = 0$. Subtracting the Smarr relation \eqref{eq:smarr} for the extremal spacetime from the subextremal version and dividing by $\kappa$, in the limit $\kappa \to 0$ we recover \eqref{eq:Smarr} upon setting
\begin{equation}
    \frac{\omega_i}{\alpha} = \lim_{\kappa \to 0}\frac{\Omega_i - \Omega_i^{ext.}}{\kappa} \hspace{.8cm}\text{ and }\hspace{.8cm} \frac{\mu}{\alpha} = \lim_{\kappa \to 0} \frac{\Phi_\mathcal{H} - \Phi_{\mathcal{H}}^{ext.}}{\kappa}.
\end{equation}
This is consistent with the interpretation of $\alpha$ in \cite{Colling:2025dub} as an extremal analogue of the surface gravity. The mass term in \eqref{eq:smarr} does not contribute to the limit since $M - M^{ext.} = O(\kappa^2)$, as argued in \cite{Johnstone:2013ioa}. 

\section{Near-horizon equations}\label{sec:solve}

In this section we solve the near-horizon equations introduced in \S\ref{ehorsec}. Without any symmetry assumptions, in \S \ref{sec:solvematter} we find an expression for the magnetic field in terms of the remaining horizon data and reduce the matter equations of motion to a single equation for the electric field. In order to make further progress we then restrict to solutions with $U(1) \times U(1)$ symmetry. After introducing the general Ansatz in \S\ref{sec:u1ans}, we show in \S \ref{sec:U1} that the equations of motion reduce to three coupled second order ODEs and a single first order constraint for the functions $\vert K \vert^2, \Gamma$ and $\chi$. We deduce that any family of solutions with fixed $\lambda$ that is globally defined on $S^3$ can depend on at most five continuous parameters. Finally, in \S \ref{sec:guess}, we assume that the electric field $\chi$ is a non-zero constant and completely solve the remaining near-horizon equations. 

\subsection{Solving the matter equations}\label{sec:solvematter}
The matter equations (\ref{nhe2}, \ref{nhe3}) for a rotating horizon can be reduced to a single scalar equation for the co-rotating electric field $\chi$ without any symmetry assumptions. Wherever $K$ is non-zero, this equation reads
\begin{multline}\label{eq:matfinal}
        \vert K \vert^2 \nabla^a\left(\frac{\Gamma}{\vert K \vert^2}\nabla_a\chi\right) \\
        + \frac{1}{3\Gamma\vert K \vert^2}\left[3\vert K \vert^4\chi + \left(c\sqrt{3}+2\lambda \chi^2\right)\left(\b \sqrt{3} + 4\lambda \vert K \vert^2\chi + 4c\chi\sqrt{3} + \tfrac{8}{3}\lambda \chi^3\right)\right]= 0.
\end{multline}
Here $\beta$ and $c$ are constants defined below. The magnetic field $B$ is determined algebraically by
\begin{equation}\label{eq:Bsol}
    B = \frac{K}{\vert K \vert^2} \wedge \text{d}\chi + \frac{1}{\Gamma \vert K \vert^2}\left(c + \frac{2\lambda}{\sqrt{3}}\chi^2\right)(\star K).
\end{equation}
Here we use $K$ to denote both the Killing vector and the 1-form $\gamma$-dual to it. Observe that $B$ is closed as a consequence of $\mathcal{L}_K B = 0$. Equation~\eqref{eq:Bsol} follows from the relations
\begin{equation}\label{eq:hookKB}
    \iota_K B = \text{d}\chi, \hspace{.8cm} \iota_K \star B = \frac{1}{\Gamma}\left(c + 
    \frac{2\lambda}{\sqrt{3}}\chi^2\right).
\end{equation}
The first equation is just (\ref{nhe3}), whereas the second one follows from (\ref{nhe2}, \ref{nhe3}) and invariance under the Killing vector $K$:
\begin{align*}
    0 = \mathcal{L}_K (\Gamma\star B) &= \text{d}\iota_K(\Gamma \star B) + \iota_K\text{d}(\Gamma \star B)\\
    &= \text{d}\iota_K(\Gamma \star B) - \iota_K\left(\chi \star K  + \frac{4\lambda}{\sqrt{3}}\chi B\right)\\
    &= \text{d}\left(\Gamma \iota_K \star B - \frac{2\lambda}{\sqrt{3}}\chi^2\right).
\end{align*}
This establishes the existence of a constant $c$ as in \eqref{eq:hookKB}. There is one remaining matter equation for $\chi$, which may be obtained by inserting \eqref{eq:Bsol} into \eqref{nhe2}. The result is
\begin{equation} \label{chi}
    \vert K \vert^2 \nabla^a\left(\frac{\Gamma}{\vert K \vert^2}\nabla_a\chi\right) + \frac{1}{\Gamma}\chi\left(\vert K \vert^2 + \frac{4\lambda c}{\sqrt{3}} + \frac{8\lambda^2}{3}\chi^2\right) + \frac{c\sqrt{3} + 2\lambda \chi^2}{\vert K \vert^2\sqrt{3}}\star(K \wedge \text{d}K) = 0.
\end{equation}
Let us write $Y = \star(K \wedge \text{d}K)$ for the twist of $K$. We make use of the general identity
\begin{equation*}
    \text{d}Y = 2\star (K \wedge \iota_K\text{Ric})
\end{equation*}
together with the gravity equation \eqref{nhe1} to obtain an expression for $Y$. Note that only the terms involving the Hessian $\nabla_a\nabla_b \Gamma$ and $B_{ac}B_b^{\:\:c}$ in the Ricci tensor will contribute. Explicitly, using index-free notation,
\begin{align*}
    \frac 12\text{d}Y &= \frac{1}{\Gamma}\star(K \wedge \iota_K \text{Hess }\Gamma) - 2\star(K \wedge \iota_{\iota_K B}B) \\
    &= \frac{1}{2\Gamma}\star (K \wedge \iota_{\nabla \Gamma}\text{d}K) - 2\star (K \wedge \iota_{\nabla \chi}B) \\
    &= -\frac{1}{2\Gamma}\star \iota_{\nabla \Gamma}(K \wedge \text{d}K) - \frac{2}{\Gamma \vert K \vert^2}\left(c + \frac{2\lambda}{\sqrt{3}}\chi^2\right)\star(K \wedge \iota_{\nabla \chi}\star K) \\
    &= -\frac{1}{2\Gamma}Y\text{d}\Gamma + \frac{2}{\Gamma}\left(c + \frac{2\lambda}{\sqrt{3}}\chi^2\right)\text{d}\chi.
\end{align*}
It follows that there is a constant $\b$ such that 
\begin{equation}\label{eq:Q2}
    \Gamma \star(K \wedge \text{d}K) = \b + 4c\chi + \frac{8}{3\sqrt{3}}\lambda \chi^3.
\end{equation}
Substituting this result into \eqref{chi}, we obtain \eqref{eq:matfinal}.

\subsection{Horizons with \texorpdfstring{$U(1)\times U(1)$}{U(1) x U(1)} symmetry} \label{sec:u1ans}
All known rotating solutions to the near-horizon equations \eqref{bd:nhe}, including the ones presented in this work, are invariant under an isometric $U(1) \times U(1)$ action. Below we introduce the most general cohomogeneity-one Ansatz for the near-horizon geometry compatible with these rotational  symmetries following  \cite{KLvac, Hollands:2010bf}. Previously known solutions relevant to this work are presented in this Ansatz in Appendix \ref{Aknown}.

Let $m_1,m_2$ denote the Killing vectors generating the $U(1) \times U(1)$ action. We introduce coordinates $\phi_1,\phi_2 \in [0,2\pi)$ along the orbits of $m_1,m_2$ such that $m_i = \partial_{\phi_i}$, and a coordinate $y$ on the orbit space of $m_i$ by $\text{d}y \propto \iota_{m_1}\iota_{m_2}\epsilon$, where $\epsilon$ is the volume form of $\gamma$ (note that $\iota_{m_1}\iota_{m_2}\epsilon$  is closed). Then $y$ takes values in a closed interval whose endpoints correspond to fixed points of the $m_i$. We fix the proportionality constant such that $y \in [0,1]$ and $(y,\phi_1,\phi_2)$ is a positively oriented coordinate chart. If $H \cong S^3$ different Killing fields vanish at the endpoints of $y$, whereas for a ring $S^2 \times S^1$ the same Killing vector vanishes at both endpoints. The induced metric $\gamma$ in the coordinates $(y,\phi_1,\phi_2)$ reads
\begin{equation}
    \gamma = \frac{a^2}{\text{det }f}\text{d}y^2 + f_{ij}(y)\text{d}\phi^i\text{d}\phi^j
\end{equation}
for some constant $a$. It is well known (see e.g. \cite{Hollands:2010bf, Kunduri:2011zr}) that the determinant of the Gram matrix $f_{ij}$ satisfies 
\begin{equation}
    \frac{\text{d}^2}{\text{d}y^2}(\Gamma\text{det }f) = -2\alpha a^2.
\end{equation}
In the present context this follows from contracting the $\phi_i\phi_j$-components of the gravity equation~(\ref{nhe1}) with the inverse matrix $f^{ij}$ and using the expression (\ref{alpha}) for $\alpha$. Since $\text{det } f$ must vanish at the endpoints of $y$, regularity requires 
\begin{equation}
    \text{det }f = \frac{\alpha a^2y(1-y)}{\Gamma}.
\end{equation}
We further find it convenient to normalise $\Gamma$ such that $\alpha = 4$. The full near-horizon geometry \eqref{bd:nh2} becomes
\begin{subequations}\label{bd:glob}
\begin{align}
    g &= \Gamma(y)[2\text{d}v\text{d}\rho - 4\rho^2\text{d}v^2] + \frac{\Gamma(y)}{4y(1-y)}\text{d}y^2 + f_{ij}(y)(\text{d}\phi^i + \omega^i\rho\text{d}v)(\text{d}\phi^j + \omega^j\rho\text{d}v), \label{gglob}\\
    F &= \chi(y)\text{d}v \wedge \text{d}\rho + B_i(y)\text{d}y \wedge (\text{d}\phi^i + \omega^i\rho\text{d}v). \label{fglob}
\end{align}
\end{subequations}
Smoothness of \eqref{bd:glob} requires  $\Gamma, f_{ij}, \chi$ and $B_i$ to extend to smooth functions of $y$ on some open interval containing $[0,1]$. To see this, note that a function $h(s)$ of the proper distance 
\begin{equation*}
s = \left\vert\int_{y_0}^y \sqrt{\gamma_{yy}} \text{ d}y\right\vert
\end{equation*}
from $y_0 \in \{0,1\}$ is smooth at $s = 0$ if and only if it can be extended as a smooth function $\tilde{h}(t)$ of $t = s^2$ to an interval containing zero. Since $\gamma_{yy}$ has simple poles in $y$ at the endpoints, $t(y)$ is smooth and d$t/\text{d}y$ approaches a non-zero constant as $y \to y_0$. Hence we can extend $h$ as a smooth function $\tilde{h}$ of $y$ across $y = y_0$.

Finally, to ensure regularity it remains to remove potential conical singularities at $y= 0,1$. For the case where $H$ has spherical topology we may assume without loss of generality that $m_1$ vanishes at $y = 1$ and $m_2$ vanishes at $y = 0$. We then impose
\begin{subequations} \label{bd:cond}
\begin{align}
f_{11}(y\to 1) &= \Gamma(1)(1-y) + \mathcal{O}\left((1-y)^2\right), \label{conda}\\
f_{22}(y\to 0) &= \Gamma(0)y + \mathcal{O}\left(y^2\right), \label{condb}\\
f_{12}(y\to0) &= \mathcal{O}(y), \quad f_{12}(y\to1) = \mathcal{O}(1-y).
\end{align}
\end{subequations}
If instead $H$ has topology $S^2 \times S^1$ with $m_1$ vanishing at both endpoints, the conditions \eqref{bd:cond} should both be imposed on $f_{11}$ and there is no requirement for $f_{22}$.

\subsection{Equations for \texorpdfstring{$\vert K \vert^2, \Gamma, \chi$}{K2, Gamma, chi}}\label{sec:U1}

Before imposing the equations of motion, it is useful to introduce linear combinations $\varphi_1,\varphi_2$ of $\phi_1,\phi_2$ such that $K = \partial_{\varphi_1}$ and det $\gamma = 1$ in the coordinates $(y,\varphi_1,\varphi_2)$. Note that $\varphi_{1,2}$ need not be $2\pi$-periodic and $K$ may not have closed orbits. The horizon data is then locally of the form

\begin{subequations}\label{eq:3Dans}
    \begin{gather}
    \gamma = \frac{ \Gamma(y) \dd y^2}{4 y (1-y)} + Q_1(y) \left( \dd \varphi_1 + Q_2(y) \dd \varphi_2 \right)^2 + \frac{4 y (1-y)}{\Gamma(y) Q_1(y)} \dd \varphi_2^2, \\
    B = \dd C, \quad C = - \chi(y) \, \dd \varphi_1 + Q_3(y) \, \dd \varphi_2,
    \end{gather}
\end{subequations}
with unknown functions $Q_1(y)$, $Q_2(y)$, $Q_3(y)$, $\Gamma(y)$, $\chi(y)$ to be solved for, where $Q_1 = \vert K \vert^2$. This Ansatz automatically satisfies (\ref{nhe3}). For ease of presentation, we suppress the $y$ dependence of the functions in the equations below. 

The equations (\ref{eq:hookKB}, \ref{eq:Q2}) allow us to determine $Q_2$ and $Q_3$ from the other three functions:
\begin{subequations}\label{eq:solveQ2Q5}
    \begin{gather}
Q_2' = -\frac{ \beta  + 4c \chi}{Q_1^2 \Gamma} - \frac{8 \l \sqrt{3} \chi^3}{9 Q_1^2 \Gamma}, \\
Q_3' = - Q_2 \chi' - \frac{ c}{Q_1 \Gamma} - \frac{2\l \sqrt{3} \chi^2}{3 Q_1 \Gamma}.
\end{gather}
\end{subequations}
Here a prime denotes a $y$-derivative. Plugging these expressions into the remaining equations (i.e. the only unsolved matter equation \eqref{eq:matfinal} and the gravity equations \eqref{nhe1}), we find that all terms involving $Q_2$ or $Q_3$ can be eliminated. This can be traced down to the fact that only the combinations $Q_2'$ and $Q_3' + Q_2 \chi'$ are invariant under the remaining gauge and coordinate freedom. As a result, any integration constant arising from solving \eqref{eq:solveQ2Q5} for $Q_2$ and $Q_3$ can be absorbed into gauge and coordinate transformations. It follows that any solution is uniquely determined by the functions $Q_1,\Gamma,\chi$ together with the constants $\beta$ and $c$.

The functions $\Gamma$ and $\chi$ are constrained by the second order equations \eqref{alpha} and \eqref{eq:matfinal} respectively. One could similarly derive an equation\footnote{This calculation is based on the identity $\frac 12\Delta \vert K \vert^2  = \vert \nabla K \vert^2 - \text{Ric}(K,K)$ satisfied by any Killing vector. The term $\vert \nabla K\vert^2$ can then be computed using \eqref{eq:Q2} and the term $\text{Ric}(K,K)$ using \eqref{nhe1} and \eqref{eq:Bsol}.} for $\Delta \vert K \vert^2$ by contracting \eqref{nhe1} twice with $K$. However, there exists a simpler equation involving $Q_1''$ which can be derived as in the vacuum case (see \cite{KLvac}) by combining the $yy$ and $\varphi_1\varphi_1$ components of the gravity equations \eqref{nhe1} with \eqref{alpha}. This leads to the following system of equations
\begin{subequations}\label{bd:es}
\be\label{eq:e1}
    \frac{\text{d}}{\text{d}y}\left(\Gamma Q_1'\right) + Q_1 \left(2 \Gamma'' - \frac{\Gamma'^2}{\Gamma}\right) + 4 \, \Gamma \chi'^2 = 0, 
\ee
\begin{equation}\label{eq:e2}
    \Gamma\frac{\text{d}}{\text{d}y}\left(\frac{P\Gamma'}{\Gamma}\right) - Q_1 + 8\Gamma - \frac 83\chi^2 - \frac{4(\sqrt{3}c + 2\lambda\chi^2)^2 + 12 P\Gamma\chi'^2}{9Q_1} = 0,
\end{equation}\vspace{.1cm}
\begin{equation}\label{eq:e3}
    Q_1\frac{\text{d}}{\text{d}y}\left(\frac{P\chi'}{Q_1}\right) + \frac{1}{3Q_1\Gamma}\left[3Q_1^2\chi + \left(c\sqrt{3} + 2\lambda\chi^2\right)\left(\beta\sqrt{3}+4\lambda Q_1\chi + 4c\chi\sqrt{3} + \tfrac 83\lambda \chi^3\right)\right] = 0.
\end{equation}
\end{subequations}
Here $P = P(y)  = \a y (1-y) = 4 y (1-y)$. Note that $P$, $\chi$, $Q_1$, and $\Gamma$ are functions of $y$, whereas $\l$, $c$, and $\b$ are constants. In addition, the functions satisfy a first order constraint 
\begin{align} \label{eq:cons}
Q_1^3 + \Gamma PQ_1'^2& - 16Q_1^2\Gamma - 2Q_1^2P'\Gamma' + \frac{P\Gamma'^2Q_1^2}{\Gamma} + Q_1Q_1'(P\Gamma' - P'\Gamma) \nonumber \\
&+ 4Q_1\left(c + \frac{2\lambda}{\sqrt{3}}\chi^2\right)^2 + \left(\beta + 4c\chi + \frac{8}{3\sqrt{3}}\lambda \chi^3\right)^2 + 4Q_1(P\Gamma \chi'^2 + \chi^2 Q_1) = 0.
\end{align}
This condition comes from solving for second derivatives of $Q_1, \Gamma, \chi$ from \eqref{bd:es} and substituting this back into the equations of motion. All matter equations are solved, as expected, and all gravity equations degenerate to this constraint. The derivative of the left-hand-side of \eqref{eq:cons} vanishes provided that \eqref{bd:es} are satisfied. Hence, to construct the general local rotating solution to the equations of motion with $U(1) \times U(1)$ symmetry, it suffices to solve \eqref{bd:es} and impose \eqref{eq:cons} at a single point.

Let us compare this reduced system of equations to the corresponding system in the vacuum case. For vacuum solutions $\chi$ and $c$ vanish, so that equation \eqref{eq:e3} is automatically satisfied. The constraint \eqref{eq:cons} may be viewed as defining $\beta$, which no longer appears in any of the second order equations. The remaining equations (\ref{eq:e1}, \ref{eq:e2}) for the two functions $\Gamma$ and $Q_1$ were solved in full generality in \cite{KLvac}.

For the special value of $\lambda = 1$, the equations can be recast as the equations of motion for a non-linear sigma model defined on the orbit space $H/U(1)^2$ with homogeneous target space $G_{2(+2)}/SO(4)$ \cite{G2}. This formulation was used in \cite{Kunduri:2011zr} to prove that $\Gamma$ must be a quadratic polynomial in $y$, and moreover all components of $\gamma$ and $F$ are rational functions of $y$. This is in fact also true for known solutions for other values of $\lambda$ (including the ones presented in this work). It should be noted however that the equations remain hard to solve in full generality even with a quadratic Ansatz for $\Gamma$. 

\subsection{Maximum number of solutions} \label{secmaxsol}

Although we are not able to solve the near-horizon equations in full generality, we are able to place an upper bound on the number of free continuous parameters. We show that any smooth family of solutions to \eqref{bd:es}, \eqref{eq:cons}, consisting of the functions $(\vert K \vert^2, \Gamma, \chi)$ and constants $(c,\beta)$, can be parametrised by at most five parameters\footnote{More precisely, we prove the following: suppose $U \subset \R^k$  is open and $\Phi = (\vert K \vert^2, \Gamma, \chi, c, \beta): [0,1] \times U \to \R^5$ is a continuous map such that the map $u \in U \mapsto \Phi(\cdot, u)$ is injective and assigns to each set of parameters $u$ a smooth solution to the system \eqref{bd:es}, \eqref{eq:cons} with constant $c$ and $\beta$. Then $k$ is at most five. By a smooth function from the closed interval $[0,1]$ we mean the restriction of a smooth map defined on an open interval containing $[0,1]$.}. We have already seen that such data uniquely determines the (local) solution to the near-horizon equations assuming $U(1) \times U(1)$ symmetry.

A solution to the second order system \eqref{bd:es} is determined by six boundary conditions. Only five of these are free after imposing the first order constraint \eqref{eq:cons}. Note that equations \eqref{eq:e2} and \eqref{eq:e3} have singular points at $y = 0$ and $y = 1$, so further constraints arise from requiring smoothness of the solution across these points. We will only consider regularity at $y = 0$. This does not in general guarantee regularity at $y = 1$, but it will be sufficient to deduce the upper bound on the number of solutions. Evaluating (\ref{eq:e2}, \ref{eq:e3}) at $y = 0$ yields two first order conditions 
\begin{equation} \label{cons2}
    \begin{aligned}
    4\Gamma' - Q_1 + 8\Gamma - \frac 83\chi^2  - \frac{4(\sqrt{3}c + 2\lambda\chi^2)^2}{9Q_1} &= 0, \\
    4\chi' +\frac{1}{3Q_1\Gamma}\left[3Q_1^2\chi + \left(c\sqrt{3}+2\lambda\chi^2\right)\left(\beta\sqrt{3}+4\lambda Q_1\chi + 4c\chi\sqrt{3} + \tfrac 83\lambda \chi^3\right)\right] &= 0,
\end{aligned}
\end{equation}
where all functions are evaluated at $y = 0$. It is straightforward to verify that these constraints are independent of \eqref{eq:cons}.  Note that we are assuming $Q_1(0) \neq 0$ above; if instead $K$ vanishes at the endpoint $y = 0$ we may replace the constraints by
\begin{equation} \label{cons3}
    \begin{aligned}
Q_1 &= 0, \hspace{.2cm} &
    c + \frac{2\lambda}{\sqrt{3}}\chi^2 &= 0, \\
    \beta + 4c\chi + \frac{8}{3\sqrt{3}}\lambda \chi^3 &= 0, & \hspace{.2cm} 
    4\Gamma' + 8\Gamma - \frac 83\chi^2 -\frac{16\Gamma(\chi')^2}{3Q_1'} &= 0, 
\end{aligned}
\end{equation}
again evaluated at $y = 0$. These equations are obtained from \eqref{eq:hookKB}, \eqref{eq:Q2} and \eqref{eq:e3}. We will show that a solution to \eqref{bd:es} is uniquely determined given generic values of $c$, $\beta$, $Q_1(0)$, $\Gamma(0)$, $\chi(0)$, $Q_1'(0)$, $\Gamma'(0)$, $\chi'(0)$ subject to the constraints (\ref{eq:cons}, \ref{cons2}) or \eqref{cons3}, resulting in a maximum of five parameters. The precise meaning of ``generic'' is explained below.

Let us rewrite the system \eqref{bd:es} in the first order form 
\begin{equation} \label{fosystem}
    y \mathbf{f}'(y) = \textbf{F}(y, \mathbf{f}(y)),
\end{equation}
where $\mathbf{f} = (Q_1,Q_1',\Gamma,\Gamma',\chi,\chi')$ satisfies the constraints at $y = 0$ and $\textbf{F}$ is a vector that is regular in a neighbourhood of $(0,\mathbf{f}(0))$. If the $6\times 6$ Jacobian matrix $M^i_{\:j} = \partial_{\mathbf{f}^j}\textbf{F}^i$  has no positive integer eigenvalues at $(0,\mathbf{f}(0))$ (this is generically true for our system), all derivatives of $\mathbf{f}$ at $y = 0$ can be constructed from $\mathbf{f}(0)$ by repeatedly differentiating \eqref{fosystem}. The formal power series at $y=0$ then uniquely determines the solution to \eqref{fosystem}, see \cite[Theorem 4.7]{odes} and references therein.

Non-generic cases refer to choices of $c, \b$, and $\mathbf{f}(0)$ such that $M^i_{\:j}$ has $m$ positive integer eigenvalues ($ m \neq 0$), counted with geometric multiplicity. In those cases, $m$ additional parameters are needed to specify all derivatives of $\textbf{f}$ at $y = 0$. This is because the $n$-th derivative of \eqref{fosystem} determines $\textbf{f}^{(n)}(0)$ in terms of lower order derivatives of $\textbf{f}$ up to a choice of $v^{(n)} \in \text{ker}(n I - M)$. However, we have checked that these extra parameters are compensated by $m$ independent relations between $c, \b$, and $\mathbf{f}(0)$ that arise from requiring $M^i_{\:j}$ to have $m$ positive integer eigenvalues. More precisely, solving the equations (\ref{eq:cons}, \ref{cons2}) for $Q_1'(0), \beta, \chi'(0)$ and inserting these expressions into $M^i_{\:j}$, we found that the kernel of $M^i_{\:j}$ has rank 4 (the same is true if $Q_1(0) = 0$). By inspecting the remaining quadratic factor in the characteristic polynomial of $M^i_{\:j}$, we verified that indeed requiring either root to be a positive integer always imposes an extra independent constraint. Therefore in all cases the total number of parameters is still bounded by five.

Turning our attention to the global solution on $H$, the only thing that is not fixed by the solution to the system \eqref{bd:es} is the $GL(2,\R)$ transformation relating $\varphi_1,\varphi_2$ to the $2\pi$-periodic angles $\phi_1,\phi_2$. Any regular solution must be such that at each endpoint $y = 0$ or $y = 1$, either $K = \partial_{\varphi_1}$ or $Q_2\partial_{\varphi_1} - \partial_{\varphi_2}$ vanishes. In the spherical case the Killing vectors vanishing at the two endpoints are different and correspond to $\partial_{\phi_1}$ and $\partial_{\phi_2}$ up to multiplicative constants, which are fixed by requiring the absence of conical singularities as in \eqref{bd:cond}. Hence the transformation is completely determined, and we deduce that \textit{any family of rotating spherical $U(1) \times U(1)$ symmetric near-horizon geometries in five-dimensional Einstein-Maxwell-Chern-Simons theory can be parametrised by at most five parameters.} This bound could be attained for $\lambda = 1$ by the horizons of the Kaluza-Klein black holes in \cite{six} which, before imposing extremality, carry six independent charges. It is however still non-trivial to verify whether five of these charges remain independent in the near-horizon limit, which we have not carried out.

For a ring $S^2 \times S^1$ the situation is different: the same Killing vector vanishes at $y = 0$ and $y = 1$, so the removal of conical singularities imposes a constraint relating data at $y = 0$ to data at $y = 1$. There is no condition on the nowhere-vanishing Killing vector tangent to the $S^1$ factor, so the size of the circle is an additional free parameter. We expect the total number of parameters to still be at most five due to the constraint coming from removal of conical singularities, but it is not obvious that this constraint is independent of the other conditions above because it involves data at both endpoints. The near-horizon geometry of the most general known extremal black string for $\lambda = 1$ indeed depends on five parameters \cite{Kunduri:2011zr}.

\subsection{Constant co-rotating electric field}\label{sec:guess}
In the rest of this work we focus on solutions with constant co-rotating electric field. Observe that for non-rotating solutions it follows from \eqref{nhe3} that $\chi$ must always be constant (see also \cite{KL} for a different argument). We will assume that $\chi$ is non-zero. We can in fact exclude the existence of non-vacuum solutions with $\chi \equiv 0$ on (a quotient of) $S^3$ without any symmetry assumptions both in the rotating and non-rotating case. Indeed, if $\chi$ vanishes then the Maxwell equation \eqref{nhe2} implies that $\Gamma B$ is co-closed. If the first cohomology group of the cross-section $H$ is trivial, there exists a globally defined function $\zeta$ such that $\text{d}\zeta = \Gamma \star B$. The requirement that $B$ is closed then becomes
\begin{equation*}
    \nabla_a(\Gamma^{-1}\nabla^a\zeta) = 0.
\end{equation*}
Multiplying this relation by $\zeta$ and integrating by parts over compact $H$ shows that the integral of $\Gamma^{-1}\vert \nabla \zeta \vert^2$ is zero, which can only happen if $\zeta$ is constant. But then $B \equiv 0$, so the entire Maxwell field vanishes. This resolves a question posed in \cite[\S 3.1.2]{KL}. On the other hand, various examples of charged near-horizon geometries with $\chi \equiv 0$ and horizon topology $S^2 \times S^1$ are known, see \cite{Kunduri:2013gce}.

Returning to $U(1) \times U(1)$ symmetric horizons, imposing $\chi(y) = \chi_0$ is a non-zero constant, equation \eqref{eq:e3} immediately implies that $Q_1(y) = q_1$ is a constant. We assume the solution is rotating so that $q_1 > 0$. Since both $\chi(y)$ and $Q_1(y)$ are constant, \eqref{eq:e1} implies
\be
2 \Gamma'' - \frac{\Gamma'^2}{\Gamma} = 0,
\ee
which has solutions of the form
\be \label{gammaps}
\Gamma(y) = [c_1 y + c_2(1-y)]^2,
\ee
where $c_1$ and $c_2$ are positive integration constants. Now we have obtained the functional form of all unknowns and solved \eqref{eq:e1}. There remain six constants $\b,c, \chi_0, q_1,c_1,c_2$. Plugging these results back into (\ref{eq:e2}, \ref{eq:e3}, \ref{eq:cons}) gives the following algebraic equations
\begin{subequations}\label{bd:fs}
\be\label{eq:f1}
    8 c_1 c_2 - \frac{4 \hat{c}^2}{3 q_1} - q_1 - \frac{8 \chi_0^2}{3}= 0, 
\ee
\be\label{eq:f2}
    9 \hat{c} \beta + 36 \hat{c}^2 \chi_0 + 9 q_1^2 \chi_0 + 4 \sqrt{3} \l \hat{c} \chi_0 (3q_1 - 4 \chi_0^2) = 0,
\ee

\be\label{eq:f3}
  3 q_1 \left[4 \hat{c}^2 + q_1 (q_1 - 16 c_1 c_2) \right] + 12q_1^2\chi_0^2 + \left(\beta\sqrt{3} + 4\hat{c}\chi\sqrt{3} - \frac{16}{3}\lambda\chi_0^3\right)^2 = 0,
\ee
\end{subequations}
where we defined $\hat{c} \equiv c + 2 \l \chi_0^2/\sqrt{3}$ for convenience. We can solve \eqref{eq:f1} and \eqref{eq:f2} for $c_1$ and $\b$, since the coefficients of these two constants cannot vanish (if $\hat{c} = 0$ then $q_1\chi_0 = 0$, contradicting our assumptions). Plugging the result into \eqref{eq:f3}, we find the final algebraic condition
\be\label{eq:final}
    \left[2 (1-2 \l) \hat{c} - \sqrt{3} q_1\right]\left[2 (1+2 \l) \hat{c} + \sqrt{3} q_1\right] q_1 \chi_0^2 - \hat{c}^2 \left(4 \hat{c}^2 - 3 q_1^2 \right) = 0
\ee
This is a constraint between three constants, $\chi_0, q_1, \hat{c}$, invariant under gauge transformations and diffeomorphisms on the horizon $H$. Note in particular that $c_2$ drops out of this final equation because it only appears in the combination $c_1 c_2$ in \eqref{bd:fs}; hence it is a free parameter. Moreover, the equation only involves $\chi_0$ through $\chi_0^2$. This can be traced to the fact that \eqref{bd:fs} is invariant under $(\b, \chi) \to (-\b, -\chi)$, which is inherited from (\ref{bd:es}, \ref{eq:cons}). Ultimately it is a consequence of \eqref{bd:fulleqs} being invariant under flipping the Maxwell field and the orientation.  

Equation \eqref{eq:final} factorises both at $\l = 0$ and at $\lambda =1$. At $\l = 0$ it becomes
\be\label{eq:l=0}
( 4 c^2 - 3 q_1^2) (c^2 - q_1 \chi_0^2) = 0.
\ee
Setting the first (second) factor to zero corresponds to the solution in \S \ref{sec:fam1} (\S \ref{sec:fam2}). Each of these is a three-parameter family. At $\l = 1$, \eqref{eq:final} becomes
\be\label{eq:l=1}
\left(\sqrt{3}q_1 + 2\hat{c}\right) \left[\hat{c}^2 \left(2 \sqrt{3} \hat{c} - 3 q_1 \right) + 3 \chi_0^2 q_1 \left(2\sqrt{3} \hat{c} +  q_1 \right) \ \right] = 0.
\ee
Again setting each factor to zero gives a three-parameter family of solutions (see \S \ref{sec:minsugra1}). Note that the branching behaviour is consistent with the homogeneous solutions in \cite{Kunduri:2013gce}.

For general $\l$, \eqref{eq:final} is a homogeneous quartic equation in $\chi_0^2, q_1, \hat{c}$ restricted to $\chi_0^2, q_1 > 0$. The homogeneity is a consequence of the symmetry $(g,F) \mapsto (k^2g, kF)$ (with $k > 0$) of the original equations of motion \eqref{bd:fulleqs}. The solutions correspond to a collection of one-parameter families ($\l$ being the parameter) of curves in $\mathbb{RP}^2$. Taken together with $c_2$ and the scaling symmetry, the solutions have three parameters per family at each $\l$. It is possible to take these parameters to be $(c_2, q_1, \hat{c})$ by solving \eqref{eq:final} for $\chi_0^2$. Instead, below we find a more democratic rational parametrisation of the quartic \eqref{eq:final} that simplifies the final expression of our solutions\footnote{We thank Khoi Le Nguyen Nguyen for pointing out this algorithm.}. This parametrisation covers all points on the quartic curve for $\l \neq 0, 1$. 

\subsubsection{Parametrising solutions}
The idea is to intersect \eqref{eq:final} with straight lines from the origin in the $(q_1, \hat{c})$ plane, which we identify with a coordinate patch of  $\RP^2$ where $\chi_0^2 = 1$. The construction fails whenever \eqref{eq:final} contains a linear factor, which happens precisely at $\l = 0,1$. The parametrisation below will miss one family of solutions for these special values of $\l$, but covers all solutions for any other $\l$. 

We first identify the singular points on the curve. Writing \eqref{eq:final} with $\chi_0^2 = 1$ as $f(q_1, \hat{c}) = 0$, singular points are points on the curve satisfying $\partial_{q_1} f = \partial_{\hat{c}} f = 0$. For $\l \neq 0, 1$ the only possibility is $(q_1, \hat{c}) = (0,0)$. At the special values $\l = 0,1$, there are additional non-isolated singular points, corresponding to the fact that the quartic is factorisable. Now consider straight lines through the origin in the $(q_1, \hat{c})$ plane,
\be
q_1 = \frac{2}{ \sqrt{3}} \k \hat{c}, 
\ee
where $\k$ is a real parameter. Such straight lines intersect the quartic curve at points satisfying 
\be
    f\left(\frac{2}{\sqrt{3}}\k \hat{c}, \hat{c}\right) = 0
    \Longleftrightarrow \hat{c}^3 \left[3 \hat{c} (\k^2-1) - 2 \sqrt{3} \k (\k + 2 \l +1)(\k + 2 \l -1) \right] = 0, 
\ee
so for each line with $\kappa^2 \neq 1$ there is a unique point of intersection apart from the origin at
\be
\hat{c} = \frac{2 \sqrt{3} \k (\k + 2 \l +1)(\k + 2 \l -1)}{3 (\k^2-1)}.
\ee
Lines with $\kappa^2 = 1$ miss the quartic if $\lambda \neq 0,1$. We can therefore use $\k$ to parametrise the full quartic curve as 
\be\label{eq:projsol}
q_1 = \frac{4 \k^2 (\k + 2 \l +1)(\k + 2 \l -1)}{3(\k^2-1)}, \quad \hat{c} = \frac{2 \sqrt{3} \k (\k + 2 \l +1)(\k + 2 \l -1)}{3 (\k^2-1)}.
\ee
Imposing $q_1 > 0$, this continuous parametrisation is cut into four disjoint regions specified by choosing signs for the different factors in the expression for $q_1$. These correspond to coloured regions in figure \ref{fig:map} (regions II and III are separated later by the line $\kappa = -2\lambda$, see below). Note that we still have the freedom to scale \eqref{eq:projsol} by a positive factor while simultaneously scaling $\chi_0^2$ (recall that we set $\chi_0^2 = 1$). We fix this freedom as follows: previously in solving the algebraic constraints we solved for $c_1$ in terms of $(c_2, q_1, \hat{c}, \chi_0^2)$; we now instead regard $c_1$ as free and solve for $\chi_0^2$ in terms of $c_1, c_2, \k$. That is, all solutions to \eqref{bd:fs} at $\l \neq 0,1$ are now parametrised by $(c_1, c_2, \k)$. 

Finally, in order to transform to the global coordinates in \S \ref{sec:u1ans} we need to find the Killing vectors $Q_2(0)\partial_{\varphi_1} - \partial_{\varphi_2}$ and $Q_2(1)\partial_{\varphi_1} - \partial_{\varphi_2}$ vanishing at the endpoints. These vectors are equal precisely when the constant 
\begin{equation} \label{tausym}
    \tau  = q_1^{-1}\left(\beta +4c\chi_0 + \frac{8}{3\sqrt{3}}\lambda\chi^3_0\right) = -\chi_0\left(\frac{q_1}{\hat{c}} +\frac{4}{\sqrt{3}}\lambda\right)
\end{equation}
vanishes, as follows from \eqref{eq:Q2}. This can only occur for the general $\lambda$ families when $\kappa = -2\l$, and we then require $\l \in (0, \frac 12)$ in order for $q_1$ to be positive. Enforcing the absence of conical singularities imposes $\Gamma(0) = \Gamma(1)$. In view of \eqref{gammaps} this can only happen if $\Gamma$ is constant, i.e. $c_1 = c_2$. The corresponding black ring solution is homogeneous and presented in the next section.

In all other cases $Q_2$ is monotonic and non-constant, so the Killing vectors with fixed points are distinct and the horizon must be spherical. This is consistent with the topological arguments in \S \ref{sec:sas}. We can always introduce coordinates $(y,\phi_1,\phi_2)$ such that $\partial_{\phi_1} \propto Q_2(1)\partial_{\varphi_1} - \partial_{\varphi_2}$ and $\partial_{\phi_2} \propto Q_2(0)\partial_{\varphi_1} - \partial_{\varphi_2}$, normalised appropriately so that \eqref{bd:cond} holds. The corresponding near horizon geometries for general $\lambda$ are presented in \S \ref{sec:genlam}. 

\section{New solutions}\label{sec:EMsol}
We present here novel three-parameter families of near-horizon geometries: two families for each special Chern-Simons coupling ($\l = 0,1$), and multiple families for other $\l$. These solutions include and generalise the homogeneous horizons classified in \cite{Kunduri:2013gce}. Novel solutions presented in this section all have constant co-rotating electric field in the sense that 
\be \label{eq:simpcond}
F_{v\rho} = \Gamma(y) \psi(y) = const. 
\ee

In \S \ref{sec:fam1} and \S \ref{sec:fam2} we present the two families at zero $\l$. Family 1 in \S \ref{sec:fam1} is connected to the static $U(1) \times U(1)$ solution \eqref{eq:static}, and family 2 in \S \ref{sec:fam2} has a vacuum limit intersecting \eqref{eq:vac} at a two-parameter slice. At $\l = 1$, we find two families, one having only a static and no vacuum limit, and the other having both. The former is presented explicitly in \S\ref{sec:minsugra1}, and the latter is contained in the general $\l$ families presented in \S \ref{sec:genlam}. 

There is one non-spherical family of solutions with $S^2 \times S^1$ topology and two parameters for any $0 < \l < 1/2$. It is homogeneous and hence already (implicitly) included in the classification in \cite{Kunduri:2013gce}. We present it here for completeness:
\begin{subequations}
\begin{gather}
    \psi(y) = \frac{1}{\g_0}\sqrt{\frac{6 (4\l^2-1)}{4 \l^2 - 3}}, \quad B = \g_0 \sqrt{\frac{3}{2(3 - 4\l^2)}} \, \dd y \wedge \dd \phi_2, \\
    \Gamma(y) = \g_0^2, \quad K = \frac{4 \sqrt{2} \g_0 \l}{L \sqrt{3 - 4\l^2}}\frac{\partial}{\partial \phi_1}, \\
    \g = L^2 \dd\phi_1^2 + \frac{\g_0^2}{4} \left(\frac{1}{y (1-y)}\dd y^2 + 4y(1-y) \dd \phi_2^2 \right),
\end{gather}
\end{subequations}
where $L >0$ and $\g_0 \neq 0$ are the two free parameters. For $\l = 0$ we recover the near-horizon geometry of the $Q = \pm P$ extremal Reissner-Nordstr\"om black hole times $S^1$, as discussed in \cite{KL}. For $\lambda = \frac 12$ the electric field vanishes and we obtain the (locally)  AdS$_3 \times S^2$ direct product solution.

\subsection{Einstein-Maxwell family 1 -- connected to static solutions}\label{sec:fam1}
This solution corresponds to setting the first factor in \eqref{eq:l=0} to zero. In terms of the Ansatz introduced in \eqref{bd:glob}, is given by
\begin{subequations}\label{eq:sol1}
\begin{gather}
    \psi(y) = -\frac{\sqrt{3 }q}{\Gamma(y)}, \quad B_1(y) = -\frac{c_1^2 \sqrt{3(c_1 c_2 - q^2)} }{2\,\Gamma(y)}, \quad B_2(y) = \frac{ c_2^2 \sqrt{3(c_1 c_2 - q^2)}}{2\,\Gamma(y)},\\
    \Gamma(y) = \left[c_1 y + c_2 (1-y)\right]^2, \quad \omega_1 = \frac{2 c_2 \sqrt{c_1 c_2 - q^2}}{c_1 q}, \quad \omega_2 = \frac{2 c_1 \sqrt{c_1 c_2 - q^2}}{c_2 q},\\
    f_{11}(y) =  \frac{c_1^2 \left[q^2 (1-y)+ c_1^2 y \right]}{\Gamma(y)} (1-y),  \quad f_{22}(y) =  \frac{c_2^2 \left[c_2^2 (1-y) + q^2 y \right]}{\Gamma(y)} y,\\
    \quad f_{12}(y) = - \frac{ c_1 c_2 (c_1 c_2 - q^2)}{\Gamma(y)}y(1-y).
\end{gather}
\end{subequations}
Note the parameter range $c_1, c_2 >0$ and $ c_1 c_2 \geq q^2 > 0$. Setting $q^2 = c_1 c_2$ recovers the known static solution \eqref{eq:static}, and setting $c_1 = c_2$ recovers the known homogeneous solution \eqref{eq:homo1}. Conserved quantities of this solution are given by
\begin{alignat}{2}\label{eq:consol1}
    S &= \frac{\pi^2}{2} c_1 c_2 |q|, \hspace{.5cm}
     &&Q = \frac{\sqrt{3} \pi q^2}{4} , \nonumber \\
    J_1 &= -\frac{\pi c_1  q^2\sqrt{c_1 c_2 - q^2}}{2 c_2}, \hspace{.5cm}
      &&J_2 = -\frac{\pi c_2  q^2\sqrt{c_1 c_2 - q^2}}{2 c_1}.
\end{alignat}
It is easy to check that the entropy satisfies
\be\label{eq:entropy1}
S = \frac{4 \sqrt{\pi}{|Q|^{3/2}}}{3^{3/4}} + \frac{3^{3/4} \pi^{3/2}}{4} \frac{|J_1 J_2|}{|Q|^{3/2}}. 
\ee
This confirms the numerical prediction of this relation in \cite{Horowitz:2024kcx}. This entropy formula is valid for the entire parameter range, but, as in \cite{Horowitz:2024kcx}, at $q^2 = \frac{c_1 c_2}{2}$ it intersects another branch of solutions that admits a vacuum limit. At the intersection point, the conserved quantities satisfy
\be \label{int}
|J_1 J_2| = \frac{16 |Q|^3}{3 \sqrt{3} \pi}. 
\ee
It may be verified that the Smarr formula \eqref{eq:Smarr} holds with $\mu = \chi + K \cdot C$ given by
\begin{equation}
    \mu = -\frac{\sqrt{3}(2q^2-c_1c_2)}{\vert q \vert}.
\end{equation}
It is interesting to look at the free energy (or more precisely, the free energy over temperature in the zero temperature limit) defined as a Legendre transform of the entropy
\be\label{eq:fdef}
I(\omega_1, \omega_2, \mu) = \left. \text{ext.}\right|_{J_1, J_2, Q} \left( S(J_1, J_2, Q) + \frac{\pi}{2}\omega_1 J_1 + \frac{\pi}{2}\omega_2 J_2 + \pi\mu Q \right).
\ee
Note that the chemical potentials $(\omega_i, \mu)$ are defined up to an overall rescaling, which we fixed by requiring $\alpha = 4$. The prefactors in \eqref{eq:fdef} would be different if we rescale the chemical potentials but the function $I$ remains unchanged. Using \eqref{eq:entropy1}, the free energy is computed to be
\be\label{eq:fsol1}
I(\omega_1, \omega_2, \mu) = \frac{4\pi^2 \mu^3}{3\sqrt{3} (4 - \omega_1 \omega_2)^2}. 
\ee
Again, due to the rescaling ambiguity of potentials, constants in \eqref{eq:fsol1}  (e.g. the $4$ in the denominator) would  change under a rescaling of $\alpha$.
It is straightforward to check that the inverse Legendre transform indeed recovers \eqref{eq:entropy1}. 

The near-horizon geometry should correspond to a solution of this extremization. That is, substituting in values of $\omega_i, \mu, J_i, Q$ from (\ref{eq:sol1}, \ref{eq:consol1}), the right-hand-side of \eqref{eq:fdef} without extremization should equal \eqref{eq:fsol1}. We have checked that this indeed holds, and explicitly the on-shell free energy \eqref{eq:fsol1} is given by
\be
I = -\frac{\pi^2}{4} (2q^2 - c_1 c_2) |q|. 
\ee
An equivalent statement is that perturbations within our family of solutions satisfy a first law (see \cite{Hajian:2013lna}), which in our conventions reads
\begin{equation} \label{firstlaw}
    \delta S + \frac{\pi}{2}\omega_1\delta J_1 + \frac{\pi}{2}\omega_2\delta J_2 + \pi\mu\delta Q = 0.
\end{equation}

\subsection{Einstein-Maxwell family 2 -- Kaluza-Klein}\label{sec:KK}\label{sec:fam2}
Our second branch of solutions is 
\begin{subequations}\label{eq:sol2}
\begin{gather}
     \psi(y)  = \frac{\sqrt{2 c_1c_2 - q^2}}{\Gamma(y)}, \quad B_1(y) = \frac{c_1^2 \sqrt{2 c_1c_2 - q^2}}{2\, \Gamma(y)}, \quad B_2(y) = -\frac{c_2^2 \sqrt{2 c_1c_2 - q^2}}{2\, \Gamma(y)},\\
    \Gamma(y) = [c_1 y + c_2 (1-y)]^2, \quad \omega_1 = \frac{4}{\omega_2} = \frac{2 c_2}{c_1},\\
    f_{11}(y) = \frac{c_1^2 \left[q^2 (1-y)+ c_1^2 y\right]}{\Gamma(y)} (1-y), \quad 
    f_{22}(y) =\frac{c_2^2 \left[q^2 y + c_2^2 (1-y)\right]}{\Gamma(y)} y,\\
    f_{12}(y) = - \frac{ c_1c_2 (c_1c_2 - q^2)}{\Gamma(y)}y (1-y), 
\end{gather}
\end{subequations}
Notice the parameter range $c_1,c_2>0$ and $2c_1c_2 \geq q^2 > 0$. Since only $q^2$ appears in the solution, we can without loss of generality take $q>0$. The thermodynamic quantities are
\begin{equation}
    S = \frac{\pi^2 c_1c_2  q}{2},\hspace{.5cm}
    J_1 = -\frac{\pi c_1^2  q}{4}, \hspace{.5cm}
    J_2 = -\frac{\pi c_2^2  q}{4},\hspace{.5cm}
    Q = -\frac{\pi}{4}q \sqrt{2 c_1c_2 - q^2}.
\end{equation}
The entropy satisfies
\be\label{eq:Ent1}
S = 2 \pi \sqrt{|J_1 J_2|}. 
\ee
This family meets the first family of solutions at $q = \sqrt{\frac{c_1c_2}{2}}$, where it satisfies \eqref{int}. For $q^2 = 2c_1c_2$ it does not reduce to the extremal Myers-Perry horizon, but rather to a different two-parameter subfamily of the general 3-parameter family of vacuum near-horizon geometries constructed in \cite{Hollands:2010bf, KLvac}. This subfamily  is homogeneous for $c_1 = c_2$, when it agrees with the horizon of the extremal Myers Perry black hole with equal angular momenta, as well as with the horizon of the extremal Kaluza-Klein black hole \cite{KK} with zero angular momentum. A region of the general $c_1 \neq c_2$ vacuum solution also intersects the fast rotating extremal KK horizons in \cite{KK}. Comparing to \cite{KLvac}, we find that in our parametrisation this happens when
\begin{equation}
   \varphi - \sqrt{\varphi} < \frac{c_1}{c_2} < \varphi + \sqrt{\varphi},
\end{equation}
where $\varphi = \frac 12(1 + \sqrt{5})$ is the golden ratio.

Observe that the induced metric $\gamma$ on $H$ is exactly the same for both families \eqref{eq:sol1} and \eqref{eq:sol2}, with the only difference being the range of the parameter $q$. In particular, the second family contains a two-parameter subfamily of non-static solutions (with non-zero angular momentum and non-vanishing matter fields), for which the horizon metric is the same as that for the static solution \eqref{eq:static}. This happens at $q^2 = c_1c_2$. Further setting $c_1 = c_2$ gives a charged and rotating solution for which $\gamma$ is the round metric on $S^3$.

A special feature of the family is that the 5D Maxwell field is null: $\vert F\vert_g^2 = \vert B \vert^2_\gamma - 2\psi^2 = 0$. This indicates that these solutions can be obtained from a 6D vacuum near-horizon geometry via Kaluza-Klein reduction along a space-like Killing vector with constant norm. The constant norm requirement ensures the absence of a dilaton field in 5D. Explicitly, it may be verified that the 6D metric 
\begin{equation} \label{g6d}
    g_6 = g + (\text{d}z + 2A)^2
\end{equation}
is Ricci-flat, where $A$ is the 5D gauge potential satisfying d$A = F$. In \cite{Hollands:2010bf} a seven-parameter family of vacuum near-horizon geometries was constructed, with horizon cross sections $S^3 \times S^1$ equipped with a $U(1)^3$-invariant metric. One parameter corresponds to the size of the $S^1$ (i.e. rescaling $z$), and requiring the norm of the Killing field tangent to the $S^1$ to be constant fixes a further two parameters. One of the remaining four parameters appears only in front of the d$z^2$ component of the metric \eqref{g6d}; remarkably, it turns out that adding an arbitrary constant times d$z^2$ preserves Ricci-flatness. From the 5D perspective this corresponds to adding a constant times $A \otimes A$ to $g$ combined with a rescaling of the matter content, which is a transformation that can be absorbed in the other three parameters. Fixing this redundancy reduces the number of parameters to three and recovers the family \eqref{eq:sol2}.

The entropy relation \eqref{eq:Ent1} can also be obtained from the 6D vacuum solutions on $S^3 \times S^1$. The full seven-parameter family satisfies the same relation \eqref{eq:Ent1}, where the subscripts $1,2$ refer to the Killing vectors tangent to the $S^3$. In particular, it does not depend on the angular momentum $J_3$ along the $z$ direction, which becomes (up to sign) the charge in 5D. The relation remains unchanged upon dimensional reduction, under which the 6D angular momenta $J_{1,2}$ become the angular momenta in 5D. Moreover, the fact that the preferred Killing vector $K$ does not have components along the $S^1$ translates to the 5D solution satisfying $\mu = 0$. The charged version of the general extremal Myers-Perry horizon is expected to satisfy the same entropy relation \eqref{eq:Ent1}, as numerically checked in \cite{Horowitz:2024kcx}. In that case the Maxwell field is not null, so this  cannot be explained by the Kaluza-Klein reduction process. It is possible that our solutions and the charged Myers-Perry solutions are different slices of a more general family satisfying \eqref{eq:Ent1}. 

Given the entropy relation \eqref{eq:Ent1}, the free energy of this branch of solutions is computed as
\be
I(\omega_1, \omega_2, \mu) = \left.\text{ext.}\right|_{J_1, J_2, Q} \left(2 \pi \sqrt{J_1 J_2} + \frac{\pi}{2}\omega_1 J_1 + \frac{\pi}{2}\omega_2 J_2 + \pi \mu Q\right),
\ee
which only has solutions if 
\be\label{eq:constr}
\omega_1 = \frac{4}{\omega_2}, \quad \text{and} \quad \mu = 0.
\ee
Note that there are no such conditions between the potentials in the analogous computation for family 1 (\S\ref{sec:fam1}). In particular, the fact that $\mu = 0$ is a consequence of the entropy relation \eqref{eq:Ent1} being independent of $Q$. Indeed, these conditions are satisfied by the explicit solution \eqref{eq:sol2}. Under these conditions the free energy vanishes. Conversely, starting with vanishing free energy and imposing these constraints \eqref{eq:constr}, the following inverse Lengendre transform gives back the entropy relation \eqref{eq:Ent1}: 
\be\label{eq:ext2}
S(J_1, J_2, Q) = \left.\text{ext.}\right|_{\omega_1, \omega_2, \mu, \Lambda, \Theta} \left( - \frac{\pi}{2} \omega_1 J_1 - \frac{\pi}{2} \omega_2 J_2 - \pi \mu Q + \Lambda \mu + \Theta \left(\omega_1 - \frac{4}{\omega_2} \right) \right),
\ee
where $\Lambda$ and $\Theta$ are Lagrange multipliers. The extremum is given by 
\be
\Theta = \frac{\pi}{2} J_1, \quad \Lambda =  \pi Q, \quad \omega_1 = \frac{4}{\omega_2} = 2 \sqrt{\frac{J_2}{J_1}}, \quad \mu = 0,
\ee
and plugging these into \eqref{eq:ext2} we get back $S = 2\pi \sqrt{J_1 J_2}$ (we have without loss of generality taken $J_1 J_2>0$ in this calculation). Since the extremal charged Myers-Perry horizons are expected to satisfy the same entropy relation \cite{Horowitz:2024kcx}, our analysis holds for those horizons as well. 

\subsection{Minimal supergravity family 1}\label{sec:minsugra1}
As is the case for $\lambda = 0$, a family of solutions at $\lambda = 1$ connected to static solutions is missed by our rational parametrisation. The solution is
\begin{subequations} \label{sugrafam1}
\begin{multline}
    \gamma = \frac{\Gamma(y)}{4 y (1-y)} \dd y^2 + \frac{c_1^2 [3 c_1^2 y + q^2 (1-y)]}{3 \Gamma(y)} (1-y) \dd \phi_1^2 \\
    + \frac{2c_1 c_2 (3 c_1 c_2 - q^2)}{3\Gamma(y)} y (1-y) \dd \phi_1 \dd \phi_2 + \frac{c_2^2[3 c_2^2 (1-y) +  q^2 y]}{3 \Gamma(y)^2}y\dd \phi_2^2,
\end{multline}
\begin{equation}
    \Gamma(y) = [c_1 y + c_2 (1-y)]^2, \quad \psi(y) = -\frac{q}{\Gamma(y)}, \quad B = \frac{\sqrt{3 c_1 c_2 - q^2}}{2 \Gamma(y)} \dd y \wedge \left(c_1^2 \dd \phi_1+ c_2^2 \dd \phi_2 \right),
\end{equation}
\begin{equation}
    \omega_1 = \frac{2 c_2 \sqrt{3 c_1 c_2 - q^2}}{c_1 q}, \quad \omega_2 =  -\frac{2 c_1 \sqrt{3 c_1 c_2 - q^2}}{c_2 q},
\end{equation}
\end{subequations}
where $c_1, c_2 >0 $ and $0< q \leq \sqrt{3 c_1 c_2}$. This solution reduces to the $U(1)^2$-symmetric static solution \eqref{eq:static} (which solves the near-horizon equations for any $\l$) at $q^2 = 3 c_1 c_2$, but it does not have a vacuum limit. The homogeneous $(c_1 =c_2$) case of this solution agrees with one branch of the CCLP solutions with equal angular momenta, the $\hat{\delta} = -1$ branch in conventions of \cite{Horowitz:2024kcx} (see Appendix \ref{app:knownsugra}). This is also the near-horizon geometry of the BMPV black hole \cite{BMPV}. Its conserved charges are given by
\be
Q = \frac{\sqrt{3} \pi}{4} c_1 c_2, \quad  J_1 = \frac{1}{4} c_1^2 \pi \sqrt{c_1 c_2 - \frac{q^2}{3}}, \quad J_2 = -\frac{1}{4} c_2^2 \pi \sqrt{c_1 c_2 - \frac{q^2}{3}}, \quad S = \frac{\pi^2 c_1 c_2 q}{2 \sqrt{3}}.
\ee
It has a simple entropy relation
\be \label{Scclp}
S = \sqrt{\frac{16 \pi}{3 \sqrt{3}} Q^3 + 4 \pi^2 J_1 J_2} \, .
\ee 
Equation \eqref{Scclp} also holds for the CCLP horizons, even though these solutions do not agree with \eqref{sugrafam1} in the non-homogeneous case.

The family analogous to family 2 in \S\ref{sec:fam2} for minimal supergravity is included in the general parametrisation below in \S \ref{sec:genlam}, so we will not write it explicitly here. A major difference from the pure Einstein-Maxwell solution is that it has both vacuum and static limits. 

\subsection{General Chern-Simons coupling}\label{sec:genlam}
There is a three-parameter family of solutions for any $\l$ extending the Kaluza-Klein family in \S \ref{sec:fam2}. The solutions in fact consist of several disjoint regions. The entropy relation may differ between regions, but the fields have the same expression in terms of parameters $(c_1, c_2, \k, \l)$. Here $\k$ takes values in different ranges for different branches. Below we first present the general solution and then separate each region more carefully. The metric is
\begin{subequations} \label{eq:lsol}
\begin{multline}\label{eq:lsolution}
\gamma = \frac{\Gamma}{4 y (1-y)}\dd y^2 + \frac{c_1^3 \left[c_1 \, F(\k,\l) \, y + c_2 \, H(\k,\l)\, (1-y)\, \right]}{F(\k, \l)\, \Gamma}(1-y) \dd \phi_1^2 \\
- \frac{2 \, c_1^2\, c_2^2\, G(\k, \l)}{F(\k,\l)\, \Gamma}\, y (1-y)\, \dd \phi_1 \dd \phi_2 + \frac{c_2^3 \left[c_2 F(\k, \l)\, (1-y) + c_1\, y H(\k,\l)\right]}{F(\k, \l)\, \Gamma} \, y \, \dd \phi_2^2, 
\end{multline}
where
\begin{equation}
\begin{aligned}
    \Gamma(y) &= \left[c_1 y + c_2(1-y) \right]^2,\\
    F(\k, \l) &= \k^4 + 4 \l \k^3 + 2 (1 + 2 \l^2) \k^2 + 4 \l \k - 3 + 4 \l^2, \\
    H(\k, \l) &= 2 (\k^2 - 1)(\k + 2 \l)^2, \\
    G(\k, \l) &= (\k^2 -3)(\k + 2 \l -1) (\k + 2\l + 1).
\end{aligned}
\end{equation}
Components of the Maxwell field are given by 
\be
\psi = \epsilon \frac{\sqrt{3 c_1 c_2}}{\Gamma} \sqrt{\frac{2 (\k^2 -1)}{F(\k,\l)}},
\ee
\be
B = \epsilon \frac{\sqrt{3 c_1 c_2}}{\Gamma} \sqrt{\frac{(\k + 2 \l -1)(\k + 2 \l +1)}{2 F(\k,\l)}} \, \dd y \wedge \left(c_1^2 \dd \phi_1 + c_2^2 \dd \phi_2 \right),
\ee
where $\epsilon\in \{1,-1\}$ specifies the sign of the square root. The choice of these signs depends on the parameters $\k$ and $\l$, see table \ref{tab1}. Note that these particular signs are fixed because we have chosen the orientation to be $\dd y \wedge \dd \phi_1 \wedge \dd \phi_2$. It is a symmetry of the equations of motion to simultaneous flip the orientation and the Maxwell field.. 
The Killing vector is given by
\begin{align}
\omega_1 &=  \frac{2 \k c_2}{c_1 (\k + 2 \l)}\sqrt{\frac{(\k + 2\l -1)(\k + 2 \l +1)}{(\k^2 -1)}}, \\
\omega_2 &= - \frac{2 \k c_1}{c_2 (\k + 2 \l)}\sqrt{\frac{(\k + 2\l -1)(\k + 2 \l +1)}{(\k^2 -1)}}.
\end{align}
\end{subequations}
The solution is only real and regular in the parameter regions given in table \ref{tab1}. These correspond to the regions in which the expression \eqref{eq:projsol} for $|K|^2$ is positive. The solution for $\kappa = -2\lambda$ with $\lambda \leq \frac 12$ is only regular on a ring $S^2 \times S^1$ and was presented at the beginning of this section. A plot of the parameter regions in the two dimensional parameter space $(\k, \l)$ is presented in figure \ref{fig:map}. 
\begin{table}[h]
\centering
\begin{tabular}{|c|c|c|c|}
\hline
I & $1 < \k$  & $\text{any} \, \, \l$ & $-$ \\
\hline
II & $\text{Max}(-2\l,-1) < \k < 1-2\l $ & $\l < 1$ & $-$ \\
\hline
III & $-1 < \k < -2 \l$ & $\l <1/2$ & $+$\\
\hline
IV & $1 - 2 \l < \k < -1 $ & $\l > 1$ & $-$\\
\hline
V & $\k < -1 - 2\l$ &  $\text{any}\, \,  \l$ & $+$\\
\hline
\end{tabular}
\caption{\textbf{Parameter range} of our solutions and the corresponding choice of signs of the matter fields. There are five disjoint regions, labelled I to V as in the first column, with ranges of $\k$ and $\l$ in the second and third column, respectively. The sign $\epsilon$ is given in the fourth column.}
\label{tab1}
\end{table}

\begin{figure}[h]
    \centering
    \includegraphics[width=.8\linewidth]{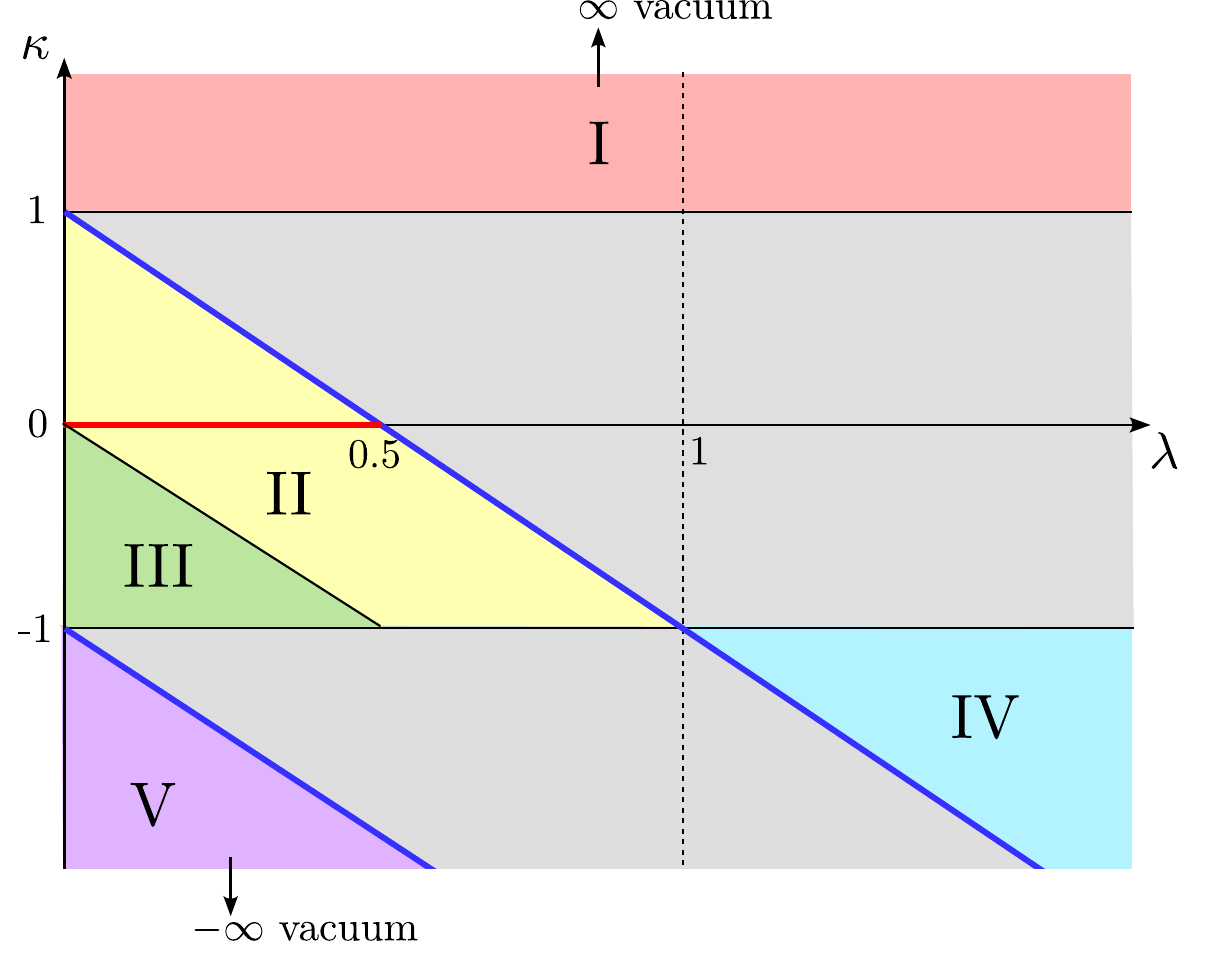}
    \caption{{\bf Parameter range} of the solutions. Regions in grey are not allowed. The two blue boundaries correspond to the same static $U(1) \times U(1)$ solution in \eqref{eq:static}. The red boundary is a new two-parameter family of static solutions. Extending either region I to $\kappa \to \infty$ or region V to $\k \to -\infty$ recovers the same vacuum limit -- the two-parameter family described in \S \ref{sec:fam2}.}
    \label{fig:map}
\end{figure}
As $\l \to 0$ all the families I, II, III, V reduce to the same solution presented in \S \ref{sec:fam2}, with $\kappa$ related to $q$ via
\begin{equation*}
    \kappa^2 = \frac{3q^2}{2c_1c_2-q^2}.
\end{equation*}
Note that in this case $\kappa \mapsto -\kappa$ is a symmetry. The solutions at $\kappa^2 = 1$ correspond to the intersection of family 1 in \S\ref{sec:fam1} with family 2 in \S\ref{sec:fam2}.

For $0 < \l < \frac 12$ and $\kappa = 0$ we encounter a static solution that to the best of our knowledge is new. It has non-zero electric and magnetic fields, and the same $U(1)^2$ symmetry for $c_1 \neq c_2$. It reduces to the extremal Reissner-Nordstr\"om-Tangherlini horizon for $c_1 = c_2$ and to the two-parameter static $U(1)^2$-symmetric solution in \cite{KL} with vanishing magnetic field for $\lambda = \frac 12$. It is interesting to note that the angular momenta $J_{1,2}$ are non-zero despite the solution being static.

The solutions in regions I and V can be connected into a single family upon replacing $\kappa$ by $\kappa^{-1}$, so that the vacuum limit is now at $\kappa = 0$. The resulting family has both a vacuum and a static limit, something which does not happen for $\l = 0$. At $\l = 1$, the homogeneous subfamily with $c_1 = c_2$ arises as the near-horizon geometry of the CCLP solutions (specifically, it is the branch $\hat{\delta} = 1$ in conventions of \cite{Horowitz:2024kcx}, see also Appendix \ref{app:knownsugra}).

\subsubsection{Entropy relations}
The conserved charges share the same expression across all regions, up to choices of signs:

\begin{equation} \label{entgen}
\begin{aligned}
    S = (c_1 & c_2)^{3/2} \pi^2 |\k + 2 \l| \sqrt{\frac{\k^2 - 1}{2\, F(\k, \l)}},\quad Q = -\frac{\sqrt{3} c_1 c_2 \pi \, \X(\k, \l)}{2 \,F(\k, \l)},\\
    J_1 &= -\sigma c_1^2 \pi \Sigma(\k,\l) \sqrt{\frac{ c_1 c_2 (\k + 2 \l - 1) (\k + 2\l +1)}{8 F(\k, \l)^3}}, \\ 
    J_2 &= \sigma c_2^2 \pi \Sigma(\k,\l) \sqrt{\frac{ c_1 c_2 (\k + 2 \l - 1) (\k + 2\l +1)}{8 F(\k, \l)^3}},
\end{aligned}
\end{equation}
where the sign $\sigma \in \{-1,1\}$ is positive for regions I, III, IV, and negative for regions II, V, and  
\be
\begin{aligned}
    \X(\k, \l) &= \k^3 + 3\l \k^2 + (4 \l^2-1)\k + \l (4 \l^2 -3), \\
\Sigma(\k, \l) &= \k^5 + 4 \l \k^4 + 2 (1 + 2 \l^2)\k^3 + 4 \l \k^2 + (4 \l^2-3) \k + 8 \l (\l^2 - 1) . 
\end{aligned}
\ee

At $\l = 1$, the two three-parameter families in region I and V both satisfy a simple entropy relation agreeing with \eqref{Scclp}, 
\be
\left.S^{(I,V)}\right|_{\l = 1} = \sqrt{\left| \frac{16 \pi}{3 \sqrt{3}} (Q^{(I,V)})^3 + 4\pi^2 J_1^{(I,V)} J_2^{(I,V)} \right|}.
\ee

For general $\lambda$ we do not have a closed form expression for the entropy relations. For a fixed ratio $\vert J_1/ J_2 \vert = c_1^2 / c_2^2$, we may plot the curves \eqref{entgen} as a function of $\kappa$ in a plane with the (normalised) entropy $S/ \vert Q \vert^{\frac 32}$ and angular momentum $\vert J_1\vert / \vert Q\vert^{\frac 32}$ on the axes. The corresponding curves for regions III and V at $\l = 0.01$ are plotted in figure \ref{fig:two-plots1} for a (restricted) range of $\kappa$. There is perfect qualitative agreement with the plots in \cite[figure 2]{Horowitz:2024kcx} based on numerical solutions. The orange curve has further turning points outside the plotted region.  There are two further branches (regions I and II) at $\l = 0.01$, whose entropy plots are presented in figure \ref{fig:two-plots}. Both branches become degenerate (zero entropy) at 
\begin{equation*}
    \frac{\vert J_1\vert}{\vert Q \vert^{\frac 32}} = \frac{2c_1}{3^{\frac 34}c_2\sqrt{\pi\lambda}} \approx8.083.
\end{equation*}

\begin{figure}[h!]
    \centering

    \begin{minipage}{0.48\textwidth}
        \centering
        \includegraphics[width=\linewidth]{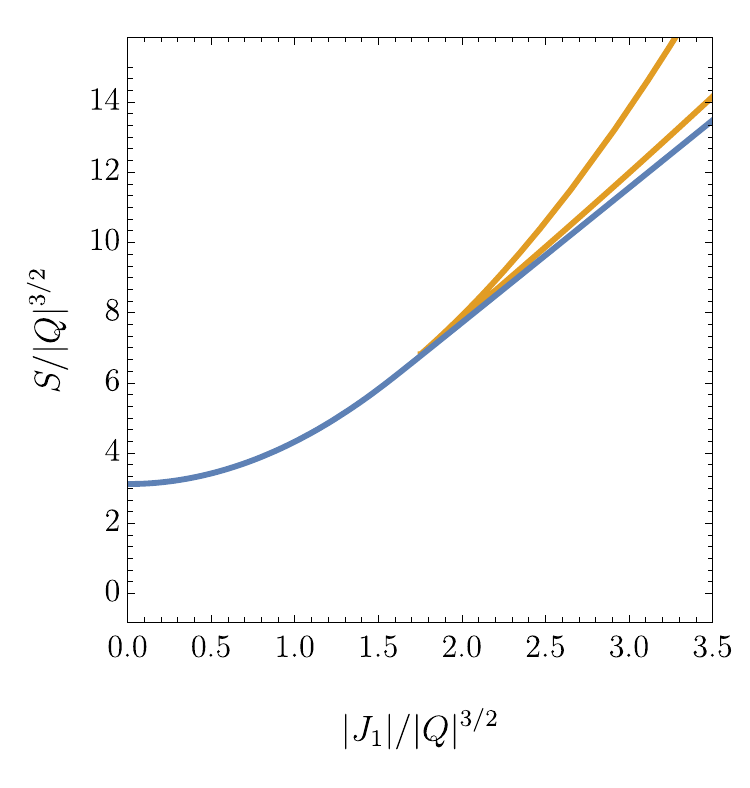}
    \end{minipage}
    \hfill
    \begin{minipage}{0.48\textwidth}
        \centering
        \includegraphics[width=\linewidth]{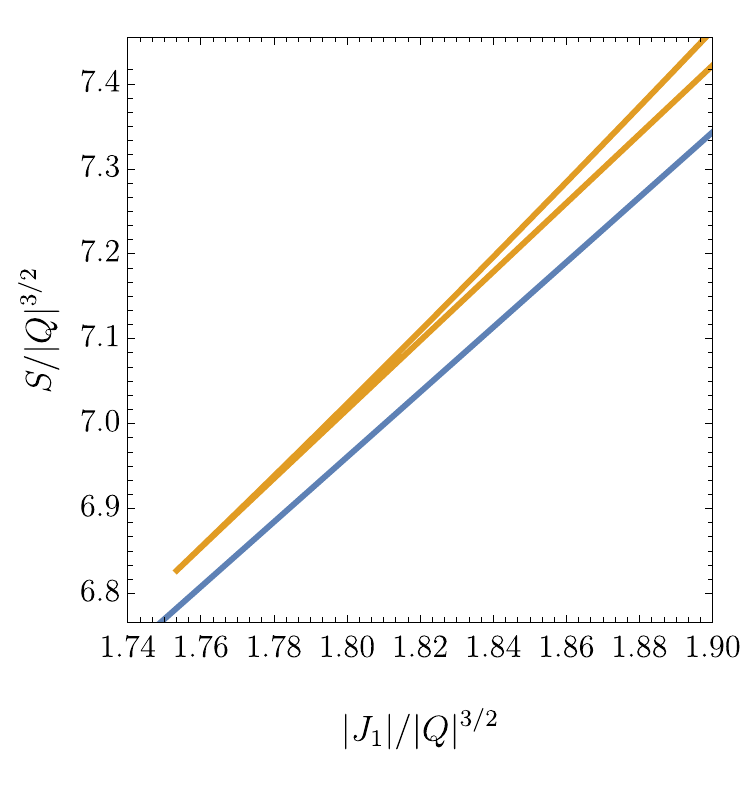}
    \end{minipage}

    \caption{\textbf{Left Panel}: normalised entropy and angular momentum for region III (orange) and region V (blue), with fixed $|J_1|/|J_2| = 3/8$ and $\l = 0.01$. \textbf{Right Panel}: a blow-up of the centre of the left panel.}
    \label{fig:two-plots1}
\end{figure}

\begin{figure}[h!]
    \centering

    \begin{minipage}{0.48\textwidth}
        \centering
        \includegraphics[width=\linewidth]{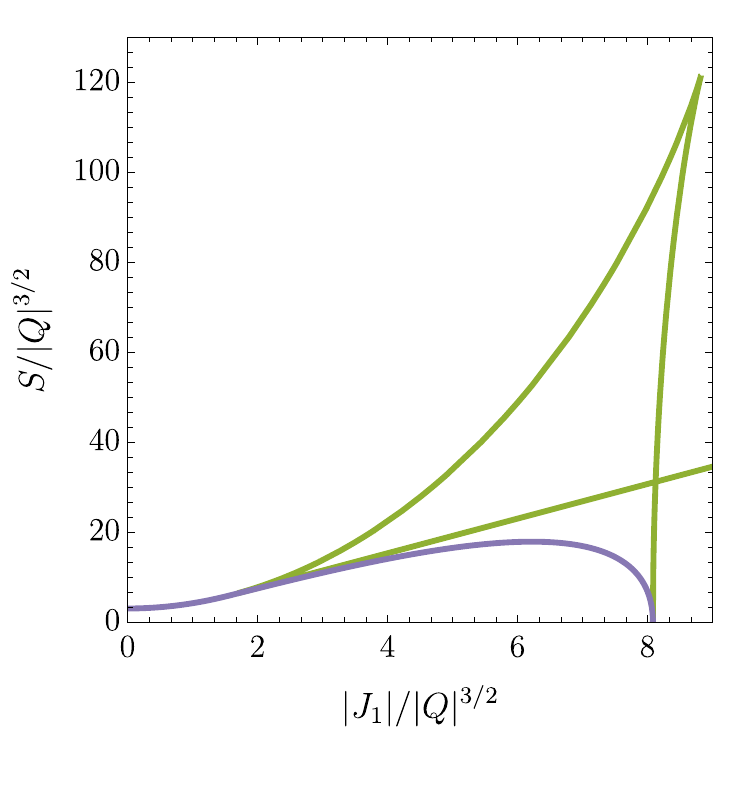}
    \end{minipage}
    \hfill
    \begin{minipage}{0.48\textwidth}
        \centering
        \includegraphics[width=\linewidth]{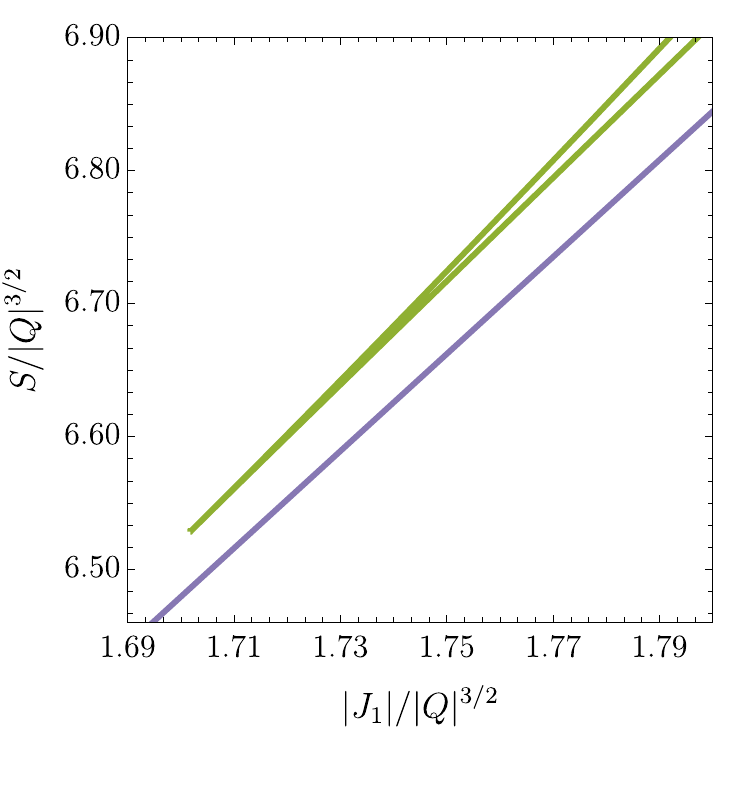}
    \end{minipage}

    \caption{\textbf{Left Panel}: normalised entropy and angular momentum for region I (green) and region II (purple), with fixed $|J_1|/|J_2| = 3/8$ and $\l = 0.01$.  The purple curve attains a static limit at zero angular momentum, but the green curve has a turning point at positive $\vert J_1\vert $. \\ \textbf{Right Panel}: a blow-up of the turning point of the green curve. }
    \label{fig:two-plots}
\end{figure}

We can further compare our results to the perturbative analysis in \cite{Horowitz:2024kcx}. For general values of $\l$, expanding in charge around the vacuum limits in regions I and V, the entropy relation has the leading order form
\be
S^{(I,V)} = 2 \pi \sqrt{\left \vert J^{(I,V)}_1 J^{(I,V)}_2\right\vert}\left[1 + \frac{4\l}{3\sqrt{3}}\frac{(Q^{(I,V)})^3}{J^{(I,V)}_1 J^{(I,V)}_2} + \mathcal{O}(\delta^4) \right]. 
\ee
Here we have assigned $Q = \mathcal{O}(\delta)$. The first correction to the vacuum entropy relation appears at order $Q^3$, as anticipated by the perturbative analysis in \cite{Horowitz:2024kcx} (even though the authors perturbed around a different background given by the Myers-Perry horizons). In regions II, IV, V, there exist a zero angular momentum limit that recovers the static $U(1) \times U(1)$ solution \eqref{eq:static}. Expanding around this limit, the entropy relations are
\begin{gather}
\frac{S^{(II,IV)}}{|Q^{(II,IV)}|^{3/2}} = \frac{4 \sqrt{\pi}}{3^{3/4}} - \frac{3^{3/4}\pi^{3/2}(1 - 2 \l)}{2 (2 - \l)} \frac{J_1^{(II,IV)} J_2^{(II,IV)}}{|Q^{(II,IV)}|^3} + \mathcal{O}(\varepsilon^3),\\
\frac{S^{(V)}}{|Q^{(V)}|^{3/2}} = \frac{4 \sqrt{\pi}}{3^{3/4}} - \frac{3^{3/4}\pi^{3/2}(1 + 2 \l)}{2 (2 + \l)} \frac{J_1^{(V)} J_2^{(V)}}{|Q^{(V)}|^3} + \mathcal{O}(\varepsilon^3),
\end{gather}
where $\varepsilon$ is proportional to the magnitude of angular momentum with $J^{(II,IV, V)} = \mathcal{O}(\varepsilon)$. In our conventions we always have $J_1 J_2<0$. The expansion in region II, IV indeed matches that in \cite{Horowitz:2024kcx} for $J_1 J_2 < 0$; the leading order term of the expansion in region V matches that in \cite{Horowitz:2024kcx} for $J_1 J_2 >0$ but with a minus sign, which we expect to be a convention issue. 

\subsubsection{Some special limits}
Apart from the vacuum and static cases, there are two additional curious families of solutions: a non-static family with zero angular momentum and a non-vacuum family with zero charge. Intuitively, this can happen for $\l \neq 0$ because the Chern-Simons term acts like a source for the Maxwell field which can carry charge and angular momentum just like ordinary matter, so that the total charge or angular momentum may vanish.

\paragraph{Zero charge but not vacuum} This happens if $\X (\k, \l) = 0$, and only occurs in region III when $\l \leq \frac{1}{2\sqrt{7}}$. In this region, there are two such solutions for each $\l$, one goes to $\k = 0$ and the other goes to $\k = -1$ as $\l \to 0$. These two solutions degenerate at $\l = \frac{1}{2\sqrt{7}}$.

\paragraph{Zero angular momentum but not static} This occurs if $\Sigma (\k, \l) = 0$ and is possible in both region III and IV. In region III there are two solutions for each $\l < \l_{max}$, with $\l_{max} \approx 0.228937$\footnote{This is the smallest positive root of $135 \l^6 - 653 \l^4 + 549 \l^2 - 27 =0$ which is a factor in the discriminant of $\Sigma(\k, \l)$.}. Similar to the zero charge case, one solution goes to $\k = 0$ and the other goes to $\k = -1$ as $\l \to 0$. These two solutions degenerate at $\l = \l_{max}$. Another branch exists in region IV if $\l \geq \l_{min}$ with $\l_{min} = 1.94251$\footnote{This is the largest positive root of $135 \l^6 - 653 \l^4 + 549 \l^2 - 27 =0$.}. There are two solutions for each $\l_{min} < \l < 2$, until one of them exits the allowed region at $(\k, \l) = (-3,2)$. For $\l > 2$ only one such solution exists for each $\l$. 

\section{Sasakian structure}\label{sec:sas}

In this section we show that there is a Sasakian structure underlying the novel solutions derived in \S\ref{sec:solve} and presented in \S\ref{sec:EMsol}. A Sasakian manifold may be defined as a Riemannian manifold $(H,h)$ equipped with a Killing vector $K$ of unit length satisfying 
\begin{equation} \label{sasak}
\nabla_a\nabla_bK^c = \delta^c_aK_b - h_{ab}K^c.
\end{equation}
Equivalently, the metric cone $(\R_{>0}\times H,\text{d}r^2 + r^2h)$ over $(H,h)$ is Kähler\footnote{The Killing field $K$ is then given by (the pullback to $H$ of) the parallel complex structure applied to the homothetic vector field $r\partial_r$.}. See e.g. \cite{Sparks} for a review.
We show that a Sasakian structure is present for rotating solutions with constant non-zero co-rotating electric field (without any symmetry assumptions), and use this to prove that any such solution must admit two commuting Killing vectors. Together with the work in the previous sections, this completes the proof of Theorem \ref{result:all}. We also discuss the generalisation of these Sasakian solutions to higher (odd)-dimensional horizons, allowing for a cosmological constant.

Recall that the matter equation \eqref{eq:matfinal} with constant $\chi \neq 0$ implies that $\vert K \vert^2$ must also be constant. Then, hooking $K$ into the hodge dual of \eqref{eq:Q2}, we deduce
\begin{equation} \label{tau}
    \text{d}K = \frac{\tau}{\Gamma}\star K,
\end{equation}
where $\tau = \vert K \vert^{-2}(\beta +4c\chi + \frac{8}{3\sqrt{3}}\lambda \chi^3)$ is a constant (compare \eqref{tausym}). Let us for the moment assume $\tau \neq 0$. By changing the orientation if necessary, we may take $\tau > 0$. We define a new metric $h$ on $H$ by rescaling the metric on the orbit space of $K$ by $
\frac{\tau}{2\Gamma}$,
\begin{equation} \label{hsak}
    h = \frac{\tau}{2\Gamma}\left(\gamma -\frac{1}{\vert K \vert^2} K \otimes K\right) + \frac{1}{\vert K \vert^2}K \otimes K.
\end{equation}
Now $K$ is also a Killing vector of $(H,h)$ and the 1-form $h$-dual to $K$ agrees with the 1-form $\gamma$-dual to $K$ (still denoted $K$). Equation \eqref{tau} becomes
\begin{equation} 
    \text{d}K = 2\star_h K.
\end{equation}
Differentiating this relation, we find that $K$ satisfies \eqref{sasak} on $(H,h)$, i.e. the manifold $(H,h,\frac{K}{\vert K \vert})$ is Sasakian. It follows that $H$ is diffeomorphic to a quotient of $S^3, \text{Nil}^3$ or $\widetilde{SL_2}(\R)$ \cite{sasaktop}, and since $H$ must have positive Yamabe invariant \cite{Lucietti:2012sa} the topology of $H$ must be spherical. In particular, the Sasakian structure is, up to a quotient, a (in general non-trivial) second type deformation of a first type deformation of the standard Sasakian structure on $S^3$ induced by the Hopf fibration \cite{sasakclas}. Our analysis splits into two cases depending on whether $\Gamma$ is constant (\S~\!\ref{sechd}) or not (\S~\!\ref{sec2kvf}).

\subsection{Second Killing field} \label{sec2kvf}
If $\nabla \Gamma \not \equiv 0$, we can use the Sasakian structure to establish the existence of a second Killing field $L$ preserving the data $(K,\Gamma,B,\psi)$. This observation relies on the fact that $\nabla K$ induces a parallel complex structure $J$ on the orbit space of $K$, and a local computation using the gravity equation \eqref{nhe1} then implies that $J(\nabla \Gamma^{\frac 12})$ is Killing on the orbit space. Explicitly, we have the following.

\begin{prop} \label{prop2kvf}
    Let $(\gamma, \Gamma,K,B,\psi)$ be a rotating solution to the near-horizon equations \eqref{bd:nhe} such that $\chi = \Gamma\psi$ equals a non-zero constant. Then the vector $L$ defined by
    \begin{equation}
        L = \Gamma^{-\frac 12}\left[\iota_{\nabla \Gamma}(\star K) -2\tau K\right]
    \end{equation}
    is a Killing vector commuting with $K$ and preserving $(\Gamma,\psi,B)$.
\end{prop}
\begin{proof}
    Observe that, in abstract index notation, 
    \begin{equation*}
        L_b = \Gamma^{-\frac 12}(\epsilon_{bcd}K^c\nabla^d\Gamma - 2\tau K_b).
    \end{equation*}
    We calculate
    \begin{align*}
        \nabla_{(a}L_{b)} &= -\frac{1}{2\Gamma}L_{(a}\nabla_{b)}\Gamma + \Gamma^{-\frac 12}\epsilon_{(b}^{\:\:\:\:cd}\nabla_{a)}(K_c\nabla_d\Gamma) \\[3pt]
        &= -\frac{1}{2\Gamma}L_{(a}\nabla_{b)}\Gamma + \frac{\tau}{2}\Gamma^{-\frac 32}\epsilon_{(b}^{\:\:\:\:cd}\epsilon_{a)ce}K^e\nabla_d\Gamma + \Gamma^{-\frac 12}K_c\epsilon_{(b}^{\:\:\:\:cd}\nabla_{a)}\nabla_d\Gamma \\[3pt]
        &= -\frac{1}{2\Gamma}L_{(a}\nabla_{b)}\Gamma - \frac{\tau}{2}\Gamma^{-\frac 32}K_{(a}\nabla_{b)}\Gamma +\Gamma^{-\frac 12}K_c\epsilon_{(b}^{\:\:\:\:cd}\nabla_{a)}\nabla_d\Gamma.
    \end{align*}
    To compute the final term, we make use of the gravity equation \eqref{nhe1}, which, after inserting the formula \eqref{eq:Bsol} for $B$, becomes
    \begin{equation} \label{gravsimp}
        \nabla_a\nabla_b\Gamma = \Gamma R_{ab} + \frac{1}{2\Gamma}(\nabla_a\Gamma)(\nabla_b\Gamma) + \frac{1}{\Gamma}K_aK_b(2\sigma^2 - \tfrac 12) - \frac{2}{3\Gamma}\gamma_{ab}(\chi^2 + 2\vert K \vert^2\sigma^2).
    \end{equation}
    Here $\sigma = \vert K \vert^{-2}(c + \frac{2\lambda}{\sqrt{3}}\chi^2)$. The terms proportional to $K_aK_b$ and $\gamma_{ab}$ do not contribute, so that 
    \begin{equation*}
        K_c\epsilon_{(b}^{\:\:\:\:cd}\nabla_{a)}\nabla_d\Gamma = \frac{1}{2\Gamma}K_c\epsilon_{(b}^{\:\:\:\:cd}\nabla_{a)}\Gamma\nabla_d\Gamma + \Gamma K_c\epsilon_{(b}^{\:\:\:\:cd}R_{a)d}.
    \end{equation*}
    Using the Ricci identity and \eqref{tau}, we can write the term involving the Ricci tensor as
    \begin{align*}
        K_c\epsilon_{(b}^{\:\:\:\:cd}R_{a)d} &= \epsilon_{(b}^{\:\:\:\:cd}\nabla_{\vert d\vert}\nabla_{a)}K_c - \epsilon_{(b}^{\:\:\:\:cd}\nabla_{a)}\nabla_d K_c \\[3pt]
        &= -\frac{\tau}{2}\nabla_{(a}\left(\Gamma^{-1}K_{b)}\right) + \tau\nabla_{(a}\left(\Gamma^{-1}K_{b)}\right) = -\frac{\tau}{2\Gamma^2}K_{(a}\nabla_{b)}\Gamma.
    \end{align*}
    Putting everything together,
    \begin{equation*}
        \nabla_{(a}L_{b)} = -\frac{1}{2\Gamma}L_{(a}\nabla_{b)}\Gamma - \frac{\tau}{2\Gamma^{\frac 32}}K_{(a}\nabla_{b)}\Gamma + \frac{1}{2\Gamma^{\frac 32}}K_c\epsilon_{(b}^{\:\:\:\:cd}\nabla_{a)}\Gamma\nabla_d\Gamma - \frac{\tau}{2\Gamma^{\frac 32}}K_{(a}\nabla_{b)}\Gamma = 0.
    \end{equation*}
    It is clear from the definition of $L$ that $\mathcal{L}_K L = 0$ and $L^a\nabla_a\Gamma = 0$. Since $\chi$ is constant and $B$ is determined by \eqref{eq:Bsol}, $L$ also leaves the matter data invariant.
\end{proof}
Note that $L$ is globally defined and not proportional to $K$ unless $\Gamma$ is constant. Indeed, since $\iota_{\nabla \Gamma} K = 0$ we can only have $\iota_{\nabla \Gamma} \star K = 0$ if $\nabla \Gamma \equiv 0$. Hence the near-horizon geometry must be invariant under a $U(1) \times U(1)$ action. If $\tau = 0$ then \eqref{tau} implies that $K$ is parallel, so $H$ cannot have spherical topology. Proposition \ref{prop2kvf} is however still valid in this case. As explained in \S\ref{sec:guess}, there are no corresponding globally regular solutions with $\Gamma$ not constant.

\subsection{Higher-dimensional Sasakian horizons} \label{sechd}

If $\chi$ and $\Gamma$ are both constant, Proposition \ref{prop2kvf} does not guarantee a second Killing field. We will show that (for $\tau \neq 0$) in this case the solution is obtained from a Sasaki-Einstein manifold. In particular, in three dimensions it must be a homogeneous three-sphere, which implies it coincides with the known equal angular momenta solutions in \cite{Kunduri:2013gce} and is contained (and explicitly parametrised) in \S \ref{sec:EMsol}. The analysis extends to Einstein-Maxwell-(Chern-Simons) theory in any odd number of dimensions and even allows for the presence of a cosmological constant $\Lambda$. We will keep the discussion general and consider the $(2n+3)$-dimensional spacetime action
\begin{equation}
    S = \int \text{d}x^{2n+3}\sqrt{-g}\:(R - 2\Lambda - F_{\mu\nu} F^{\mu\nu}) - \frac{16\lambda}{(n+1)(n+2)\sqrt{3}}\int F^{\wedge (n+1)}\wedge A.
\end{equation}
Following the steps in \S \ref{ehorsec}, the generalisation of \eqref{bd:nhe} becomes \cite{Colling:2025dub}
\begin{subequations}\label{bd:nhe2}
\begin{align}
    R_{ab} &= \frac{1}{2\Gamma^2}K_aK_b - \frac{(\nabla_a\Gamma)(\nabla_b\Gamma)}{2\Gamma^2} + \frac{1}{\Gamma}\nabla_a\nabla_b \Gamma + 2B_{ac}B_b^{\:\:c} + \tfrac{1}{2n+1}\gamma_{ab}(2\Lambda + 2\psi^2 - B_{cd}B^{cd}), \label{nhe4} \\[3pt]
    &\nabla^a(\Gamma B_{ab}) = K_b\psi + \frac{4\lambda}{\sqrt{3}}\Gamma\psi [\star (B^{\wedge n})]_b,\label{nhe5} \\
    &K^aB_{ab} = \nabla_b(\psi\Gamma). \label{nhe6}
\end{align}
\end{subequations}
Here $R_{ab}$ is the Ricci tensor of the metric $\gamma$ defined on a compact manifold $H$ of dimension $2n+1$, equipped with a Killing vector $K$ preserving a closed 2-form $B$ and functions $\Gamma,\psi$. The generalisation of the constant $\alpha$ in \eqref{alpha} is
\begin{equation}
    \alpha = \frac{\vert K \vert^2}{\Gamma} - \Gamma G = \frac{\vert K \vert^2}{2\Gamma} - \frac 12 \Delta \Gamma - \frac{2}{2n+1}\Gamma\Lambda + \frac{4n}{2n+1}\Gamma\psi^2 + \frac{1}{2n+1}\Gamma \vert B \vert^2. 
\end{equation}
For our Sasakian Ansatz we assume that $\Gamma, \psi$ are constant and that there exists a constant $b$ such that $B = \frac 12b \: \text{d}K.$ Moreover, we assume there is a positive constant $a$ such that $(H, h, \frac{K}{\vert K \vert})$ is Sasakian, where
\begin{equation} \label{hsasak}
    h = a\left(\gamma - \frac{1}{\vert K \vert^2}K \otimes K\right) + \frac{1}{\vert K \vert^2} K \otimes K.
\end{equation}
Note that, by the arguments in the previous section, any solution with $n = 1, \Lambda = 0$ and $\Gamma, \chi, \tau \neq 0$ constant is of this form. Equation \eqref{nhe6} is equivalent to $\vert K \vert^2$ being constant. We use the freedom in simultaneously rescaling $\Gamma$ and $K$ to set $\vert K \vert^2 = 1$. Equation \eqref{nhe5} can be simplified by calculating the divergence of $B$ using the expression for $\Delta K$ coming from the Sasakian condition~\eqref{sasak}. To deal with the right hand side, observe that the volume form of $\gamma$ is
\begin{equation*}
    \epsilon_\gamma = \frac{1}{a^n}\epsilon_h = \frac{1}{n!(2a)^n}K \wedge (\text{d}K)^{\wedge n} = \frac{1}{n!(ab)^n}K\wedge B^{\wedge n}.
\end{equation*}
Therefore $\star (B^{\wedge n}) = n!(ab)^nK$. We find that all terms in \eqref{nhe5} are proportional to $K$, so that the equation reduces to a relation between constants
\begin{equation} \label{consteq1}
    -2na^2b\Gamma  = \psi + \frac{4\lambda n!(ab)^n}{\sqrt{3}}\psi\Gamma.
\end{equation}
For the gravity equation \eqref{nhe4}, observe that the terms involving $\nabla \Gamma$ vanish, and 
\begin{equation*}
    B_{ac}B_b^{\:\:c} = ab^2h^{cd}\nabla_aK_c\nabla_bK_d = ab^2(h_{ab} - K_aK_b) = a^2b^2(\gamma_{ab} - K_aK_b).
\end{equation*}
Hence \eqref{nhe4} reduces to
\begin{equation} \label{gravred}
    R_{ab} = \left(\frac{1}{2\Gamma^2} - 2a^2b^2\right) K_aK_b + \frac{2}{2n+1}\gamma_{ab}\left((n+1)a^2b^2 + \Lambda +\psi^2\right).
\end{equation}
Contracting twice with $K$ and using $R_{ab}K^aK^b= \vert \nabla K \vert^2 = 2na^2$, we find a second relation between the constants $(a,b,\Gamma,\psi)$,
\begin{equation} \label{sum2n}
    2na^2 \left(1 + \frac{b^2}{2n+1}\right) = \frac{1}{2\Gamma^2} + \frac{2}{2n+1}(\Lambda + \psi^2).
\end{equation}
A computation using \eqref{sum2n} and \eqref{hsasak} shows that, in terms of $h$ and Ric$(h)$,  equation \eqref{gravred} takes the form
\begin{equation} \label{Ksasak}
    \text{Ric}(h) = \n K \otimes K + \xi h,
\end{equation}
with
\begin{equation}
    \xi + \nu = 2n, \hspace{.8cm} \xi = -2+2a + \frac{2}{a(2n+1)}\left((n+1)a^2b^2 + \Lambda + \psi^2\right).
\end{equation}
A Sasakian manifold sasifying \eqref{Ksasak} is called $K$-Einstein \cite{Boyer}. We will assume that the Sasakian structure is positive in the language of \cite{Boyer}, so that $\xi > -2$ (note this is always the case if $\Lambda \geq 0$). Then the transverse homothetic transformation
\begin{equation*}
    (h,K) \mapsto (\tilde{h}, \tilde{K}) = (\kappa h + \kappa(\kappa-1)K \otimes K, \:\kappa^{-1}K)
\end{equation*}
with $\kappa = \frac{\xi + 2}{2(n+1)} > 0$ ensures that $(\tilde{h},\tilde{K})$ is Sasaki-Einstein (see \cite[Prop. 18]{Boyer}).  To summarise, we have reduced the equations \eqref{bd:nhe2} to the relations (\ref{consteq1}, \ref{sum2n}) between constants $(a,b,\Gamma,\psi)$, together with the Sasaki-Einstein condition for $(\tilde{h},\tilde{K})$. In particular, by reversing the steps above we can construct a two-parameter family of charged and rotating extremal horizons starting from any compact Sasaki-Einstein manifold $(H,\tilde{h},\tilde{K})$.

Restricting to $n = 1$ and  $\Lambda = 0$, we are now in a position to deduce our main Theorem \ref{result:all}.

\begin{proof}[Proof of Theorem \ref{result:all}] It suffices to prove that any rotating solution to the near-horizon equations \eqref{bd:nhe} on a compact 3-manifold $H$ with constant and non-zero $\chi$ is invariant under a $U(1) \times U(1)$ action, as the equations are solved completely under these assumptions in \S \ref{sec:solve} with solutions presented in \S \ref{sec:EMsol}. If $\Gamma$ is not constant this follows immediately from the rigidity theorem in \cite{Colling:2025dub} together with Proposition~\ref{prop2kvf}. For constant $\Gamma$ there are two cases to consider depending on whether the constant $\tau$ in \eqref{tau} vanishes.

If $\tau = 0$ then $K$ is parallel and the gravity equation \eqref{gravsimp} shows that the orbit space of $K$ has constant positive curvature. Hence $(H,\gamma)$ is isometric to a quotient of $(\R \times S^2, \text{d}t^2 + \gamma_{S^2})$ with $\gamma_{S^2}$ a constant curvature metric on $S^2$ and $B$ a constant multiple of the volume form of $\gamma_{S^2}$. In particular, the solution admits two commuting Killing vectors and is (locally) isometric to the homogeneous ring solution in \S\ref{sec:EMsol}.

For $\tau \neq 0$ the solution is obtained from a Sasaki-Einstein manifold as above. Since $(H,\tilde{h})$ is positive Einstein, it must be (up to a quotient) isometric to a round $S^3$. We can view this $S^3$ as the unit sphere in $\C^2$ with $\tilde{K}$ generating the isometries $(z_1,z_2) \mapsto (e^{i\phi}z_1,e^{i\phi}z_2)$ (see \S \ref{secex}). The subgroup of $SO(4)$ commuting with these isometries is $U(2)$. The procedure to reconstruct the horizon metric $\gamma$ and the magnetic field $B$ from $\tilde{h}$ only involves $\tilde{K}$, so the $U(2)$ action also preserves $(H,\gamma, B)$. In particular, the near-horizon geometry is homogeneous and contained in the classification in \cite{Kunduri:2013gce} as well as in \S\ref{sec:EMsol}.
\end{proof}

For $n > 1$ the two-parameter family of near-horizon geometries with constant $\Gamma$ need not be homogeneous -- for example, one may take $(H, \tilde{h})$ to be given by one of the explicit non-homogeneous Sasaki-Einstein metrics on $S^2 \times S^3$ constructed in \cite{Gauntlett:2004yd}. In the special case where $\psi = B = 0$ the solutions reduce to vacuum Sasakian horizons studied in \cite{Kunduri:2012uq}.

There is no straightforward higher-dimensional generalisation of the Sasakian horizons with $n = 1$ and $\lambda = \Lambda = 0$ in this work for which $\Gamma$ is not constant. If we impose that the metric \eqref{hsak} (with $\tau$ and arbitrary constant) is Sasakian\footnote{The computations below actually only require $(H,h, K)$ to be a $K$-contact manifold, i.e. a contact metric manifold for which the Reeb vector field $K$ is Killing.} with constant $\chi$ and $B$ a constant times $\text{d}K$, then the matter equation \eqref{nhe5} and the gravity equation \eqref{nhe6} contracted once with $K$ reduce to 
\begin{equation} \label{ct}
    \star_\gamma \text{d}\star_{\gamma}(\Gamma \text{d}K) = \frac{n\tau^2}{2\Gamma}K,
\end{equation}
together with relations between constants. The Sasakian condition imposes that $\omega = \frac 12\text{d}K$ is a K\"ahler form on the orbit space of $K$ with K\"ahler metric $\hat{h}$ induced by $h$. On this orbit space \eqref{ct} implies
\begin{equation*}
    \text{d}\star_{\hat{h}}(\Gamma^{n-1}\omega) = 0.
\end{equation*}
Since $\omega$ is co-closed and non-degenerate, it follows that either $n = 1$ (corresponding to the solutions in this work) or $\Gamma$ is constant.

\subsection{Example: charged spherical horizons} \label{secex}
As an example, we set $\lambda = \Lambda = 0$ and construct charged and rotating extremal horizons in any odd dimension from the standard Sasakian structure on $S^{2n+1}$. We view $S^{2n+1}$ as the unit sphere in $\C^{n+1}$, from which it inherits the round metric $\tilde{h}$. The Sasakian vector field $\tilde{K}$ generates isometries $(z_1,\dots,z_{n+1}) \mapsto (e^{i\phi}z_1,\dots,e^{i\phi}z_{n+1})$, and the orbit space of $\tilde{K}$ is $\C\mathbb{P}^n$ equipped with the Fubini-Study metric $g_{\mathbb{CP}^n}$. We can write the metric in local coordinates where $\tilde{K} = \partial_\phi$ as 
\begin{equation}
    \tilde{h} = \tilde{K} \otimes \tilde{K} + g_{\mathbb{CP}^n}, \hspace{.8cm} \tilde{K} = \text{d}\phi +2\sigma, 
\end{equation}
where $\text{d}\sigma = \omega_{\mathbb{CP}^n}$ is the K\"ahler form.
To reconstruct the horizon metric $\gamma$, we solve the constraints (\ref{consteq1}, \ref{sum2n}) for $a$ and $b$ in terms of $\psi$ and $\Gamma$. As in the case $n =1$, the constraints factorise and there are two branches:
\begin{align*}
    \text{branch I:}& \hspace{1cm} a = \frac{\psi}{\sqrt{n(2n+1)}}, \hspace{.4cm} b = -\frac{2n+1}{2\psi\Gamma}, \\
    \text{branch II:}& \hspace{1cm} a= \frac{1}{2\Gamma\sqrt{n}}, \hspace{.4cm} b = -2\psi\Gamma.
\end{align*}
The corresponding constants $\kappa$ involved in the transverse homothetic transformation are
\begin{equation*}
    \kappa_I = \frac{4\Gamma^2\psi^2+2n+1}{4\psi\Gamma^2\sqrt{n(2n+1)}}, \hspace{.8cm} \kappa_{II} = \frac{4\psi^2\Gamma^2+1}{2\Gamma\sqrt{n}(n+1)}.
\end{equation*}
In terms of $(\psi,\Gamma)$ the remaining horizon data $(\gamma, K, B)$ reads
\begin{equation}
    \gamma = \frac{1}{a\kappa}g_{\mathbb{CP}^n} + \kappa^{-2}(\text{d}\phi + 2\sigma)^2, \hspace{.8cm} K = \kappa \frac{\partial}{\partial \phi}, \hspace{.8cm} B = b\kappa^{-1}  \omega_{\mathbb{CP}^n}.
\end{equation}
Branch I has a static limit for $\Gamma \to \infty$, which corresponds to the extremal Reissner-Nordstr\"om-Tangherlini horizon. Branch II has a vacuum limit $\psi \to 0$ which is the horizon geometry of the extremal Myers-Perry black hole with all angular momenta set equal (see \cite{MPNHG}). There are also limits of branch I as $\psi \to 0$ and of branch II as $\Gamma \to \infty$ which can be globally defined on $\mathbb{CP}^n \times S^1$.

The entropy $S$ and charges $J_\phi = J[\partial_\phi]$ and $Q$ can be computed using the formulae (\ref{Qhor}, \ref{Jhor}), with the integral now taken over the $(2n+1)$-dimensional cross-section $H$. The result is 
\begin{equation*}
    S = \frac{\pi^{n+1}}{2a^n\kappa^{n+1} n!}, \hspace{.8cm} Q = -\frac{\pi^{n}\psi}{4a^n\kappa^{n+1}n!}, \hspace{.8cm} J_\phi = -\frac{(1 - 2b\Gamma\psi)\pi^{n}}{8a^n \kappa^{n+2} \Gamma n!}.
\end{equation*}
For branch II it is straightforward to verify that
\begin{equation}
    S = \frac{2\pi}{(n+1)\sqrt{n}}\vert J_\phi\vert.
\end{equation}
Setting $n = 1$, we recover the relation \eqref{eq:Ent1} with $\vert J_1\vert = \vert J_2\vert = \frac{1}{n+1}\vert J_\phi\vert $. The solution satisfies $\vert F \vert_g^2 = 0$ for all $n$ and hence can be lifted to a vacuum near-horizon geometry with horizon topology $S^{2n+1} \times S^1$. For branch I the relation between $S, J_\phi$ and $Q$ is more complicated and given by\footnote{This formula was obtained using ChatGPT 5.4 and verified analytically by the authors.}
\be
 S = \frac{\pi |J_\phi|}{(n+1)\sqrt{n}} \left(\sqrt{z} + \frac{1}{\sqrt{z}} \right), 
\ee
where $z$ is the unique positive solution to
\be
z(1+z)^{n-1} = \frac{4^{n+1} (n+1)^{2n} n!}{\pi^n \sqrt{n}(2n+1)^{n+1/2}} \frac{|Q|^{2n+1}}{|J_\phi|^{2n}}
\ee
For $n = 1$ this reproduces the entropy relation \eqref{eq:entropy1} with equal angular momenta.

\section{Conclusion}\label{sec:concl}
In this work we have presented new charged and rotating near-horizon geometries in five dimensions with two independent angular momenta. At vanishing Chern-Simons coupling, we found two three-parameter families of solutions extending the known homogeneous solutions. A surprise is that they satisfy the same entropy relations expected for extremal charged Myers-Perry black holes \cite{Horowitz:2024kcx}, although they do not reduce to the extremal Myers-Perry horizon in the vacuum limit. We expect that our family 1 (\S\ref{sec:fam1}) coincides with the numerical solutions in \cite{Horowitz:2024kcx}, whereas our family 2 (\S\ref{sec:fam2}) might be a different slice of a more general family of solutions that contains the charged Myers-Perry horizon.

We also found generalisations of these geometries to non-zero Chern-Simons coupling. Our solutions come in five families, two of which have entropy relations that agree qualitatively with the numerics in \cite{Horowitz:2024kcx}. The perturbative expansions of the entropy relations in charge and spin also match, even though our vacuum solutions (at zero charge) are not Myers-Perry. This again suggests that there might be a more general family of solutions for general values of the Chern-Simons coupling. For the specific coupling corresponding to the bosonic part of minimal supergravity, our solutions share the same entropy relation with the extremal CCLP black holes \cite{Chong:2005hr}, but the horizons are not the same except when the angular momenta are equal.  

Our solutions contain all possible rotating near-horizon geometries for which the co-rotating electric field is a non-zero constant. The proof of this fact exploits a Sasakian structure underlying our solutions to establish the existence of a $U(1)^2$ isometric action preserving the Maxwell field. Any other solution with $U(1)^2$ symmetry (including the charged Myers-Perry horizon) can be obtained by solving the system of ODEs \eqref{bd:es} subject to the constraint \eqref{eq:cons}. We have not found any further new solutions to this system analytically, although we did establish an upper bound for the number of parameters in a smooth solution (five for a spherical horizon). We expect the derived system to be useful for future attempts at constructing more solutions.

In higher (odd) dimensions the Sasakian observation can be used to generate a two-parameter family of charged and rotating solutions starting from any Sasaki-Einstein manifold. This generalises the vacuum construction in \cite{Kunduri:2012uq}. In spacetime dimensions higher than five our Sasakian Ansatz forces $\Gamma$ to be constant, but the resulting horizons can still be non-homogeneous. It would be interesting to investigate whether a similar Sasakian construction could be carried out for near-horizon geometries in other theories such as the bosonic sector of 11D supergravity or 10D supergravities \cite{Gauntlett:2006ns,Kim:2006qu,
Couzens:2018wnk,Gauntlett:2019pqg}.

An obvious question is whether our horizons can be embedded into any known or yet to be constructed extremal black holes. In pure Einstein-Maxwell theory, it is possible that extremal charged Myers Perry horizons come in two branches (as is the case in minimal supergravity) and one of them coincides with our family 1. This family could also correspond to a rotating version of the black holes in a background electric field constructed in \cite{KL}. For family 2, a natural candidate would be a charged version of the Rasheed solutions \cite{KK}. This is of course all speculative, as none of these black hole solutions are known due to a lack of solution generating techniques in five-dimensional Einstein-Maxwell theory (see e.g. approaches in \cite{Frolov:2017kze, Ortaggio:2023rzp}). The corresponding full black holes may not exist in the full parameter range \cite{Blazquez-Salcedo:2025cpu}. Even if they do exist, they are likely to suffer from the same type of singularities as in \cite{Horowitz:2024kcx,Herdeiro:2025blx}. Within minimal supergravity, one might expect our solutions to be embeddable into the known charged Kaluza-Klein black holes \cite{six} with six independent charges. We have not been able to identify any of our non-homogeneous solutions with horizons of the slightly less general family in \cite{kk54par}.

A first step towards addressing the existence of extremal black holes corresponding to these novel geometries is to compute the scaling dimensions  \cite{Horowitz:2024kcx,Horowitz:2022mly, Horowitz:2023xyl,Horowitz:2024dch} (see also \cite{Li, maciej}). The sign of the scaling dimensions can be used to argue whether the near horizon geometry arises from an extremal black hole. Moreover, whether the scaling dimensions are positive integers can be used to diagnose singularities on the horizon. Scaling dimensions may also help to distinguish our solutions from those numerically found in \cite{Horowitz:2024kcx}. Branches of our solutions with vacuum limits should be different from those in \cite{Horowitz:2024kcx}, despite them sharing the same entropy relation. They are likely to have different spectra of scaling dimensions. However, we expect family 1 of our solutions to agree with those in \cite{Horowitz:2024kcx}, and hence also their scaling dimensions. 

\section*{Acknowledgements}
It is a pleasure to thank Maciej Dunajski, Sean A. Hartnoll, Gary T. Horowitz, James Lucietti and Jorge E. Santos for guidance and discussions throughout this project. We are also grateful to Bernardo Araneda, Cameron Gibson, Khoi Le Nguyen Nguyen and Matt Smith for very helpful discussions. This work has been partially supported by STFC consolidated grant ST/X000664/1. AC is supported by a Cambridge Trust International Scholarship. The work of AC  was partially supported by the Simons Foundation grant (award no. SFI-MPS-T-Institutes-00010825) and from State Treasury funds as part of a task commissioned by the Minister of Science and Higher Education under the project Organization of the Simons Semesters at the Banach Center - New Energies in 2026-2028 (agreement no. MNiSW/2025/DAP/491).  JL is supported by a Harding Distinguished Postgraduate Scholarship. 

\appendix

\section{Known solutions} \label{Aknown}

In this appendix we review some analytically known spherical near-horizon geometries in the theory \eqref{eq:action} relevant for our work. We will focus on $\l = 0,1$ and present the solutions in the general $U(1) \times U(1)$ symmetric Ansatz \eqref{bd:glob}. For $\l \neq 0,1$ we additionally recover known homogeneous near-horizon geometries \cite{Kunduri:2013gce} in our solutions \eqref{eq:lsol} (for $c_1 = c_2$), which we do not write explicitly here.

\subsection{Pure Einstein-Maxwell}
The known solutions in this case consist of a two parameter family of charged static solutions \cite{KL}, a three-parameter family of vacuum solutions \cite{Hollands:2010bf, KLvac}, and two two-parameter families of charged and rotating solutions with equal angular momenta \cite{Kunduri:2013gce, homog2}. The horizon cross-sections of the latter solutions have an enhanced isometry group $U(2)$ and hence are homogeneous.

\paragraph{Static} The static solutions are given by
\begin{subequations}\label{eq:static}
\begin{gather}
    \psi(y) = \frac{\sqrt{3 c_1 c_2}}{\Gamma(y)}, \quad B_i(y) = 0,\\
    \Gamma(y) = \left(c_1 y + c_2 (1-y)\right)^2, \quad \omega_i = 0,\\
    f_{11}(y) = \frac{c_1^3(1-y)}{(c_1-c_2)y+c_2}, \quad f_{12}(y) = 0, \quad f_{22}(y) = \frac{c_2^3y}{(c_1-c_2)y+c_2}.
\end{gather}
\end{subequations}
Here $c_1$ and $c_2$ are arbitrary positive constants.

\paragraph{Vacuum} The vacuum solutions have the general form
\begin{subequations}
\begin{gather}\label{eq:vac}
     \psi(y)  = 0, \quad B_i(y) = 0,\\
    \Gamma(y) = \frac{b(a + b)^3 (1-y)^2 + c^4 y(1-y) + a(a + b)^3 y^2}{c^2}, \quad\omega_1 = \frac{4}{\omega_2} = 2 \sqrt{\frac{b}{a}}, \\
    f_{11}(y) = \frac{a (a + b)^3\left(a (a + b)^3 y + c^4 (1-y)\right)}{16 c^4\, \Gamma(y)} (1-y), \\
    f_{12}(y) = \frac{a b (a + b)^6}{16 c^4\, \Gamma(y)} y (1-y), \\
    f_{22}(y) = \frac{b (a + b)^3 \left(b (a + b)^3 (1 - y) + c^4 y \right)}{16 c^4 \, \Gamma(y)} y,
\end{gather}
\end{subequations}
where $a, b, c$ are real parameters. There are two special cases relevant to us in this family of solutions. When $ (a + b)^2 = c^2$, we have the near horizon limit of the extremal Myers-Perry solution. In this case $\Gamma(y)$ is a linear function of $y$. When $ 4 a b (a + b)^6 = c^8$, the function $\Gamma(y)$ becomes a perfect square. This is the vacuum limit of a branch of our novel solutions, see \S \ref{sec:KK}. These two two-parameter subfamilies of vacuum solutions intersect at a one-parameter family with $a = b$, which agrees with the vacuum limit of a branch of homogeneous solutions (see below). The whole family shares a simple entropy relation
\be
S = 2 \pi \sqrt{|J_1 J_2|}. 
\ee

\paragraph{Homogeneous with static limit} One of the two families of homogeneous solutions include the direct product $AdS_2 \times S^3$ solution as a limit. This family reads
\begin{subequations}\label{eq:homo1}
\begin{gather}
    \Gamma(y) = c^2, \quad \psi(y) = \frac{\sqrt{3 (c^2 - l^2)}}{\Gamma(y)},\\
    \omega_1 = \omega_2 = \frac{2 l}{\sqrt{c^2 - l^2}}, \quad B_1(y) =- B_2(y) = \sqrt{3}\, l\\
    f_{11}(y) = \left[c^2 - l^2(1-y)\right] (1-y), \quad f_{12}(y) = - l^2 y (1-y), \quad f_{22}(y) = (c^2 - l^2 y)\, y,
\end{gather}
\end{subequations}
where $l$ and $c$ are real parameters with $|l| < |c|$. The $AdS_2 \times S^3$ solution is recovered by setting $l=0$ here, which corresponds to setting $c = c_1 = c_2$ in \eqref{eq:static}. 

\paragraph{Homogeneous with vacuum limit} The other family of homogeneous solutions instead has a vacuum limit recovering the vacuum solution presented above with $a = b$. These solutions have the form
\begin{subequations}
\begin{gather}
    \psi(y) = \frac{q \sqrt{3}}{\sqrt{2 (2 a^2 + q) h(a,q)}}, \quad B_1(y) = -B_2(y) = \frac{q\sqrt{3 (2 a^2 + q)}}{\sqrt{8 h(a,q)}}, \\
    \Gamma(y) = 2a^2 + q, \quad \omega_1 = \omega_2 = \frac{(2 a^2 + q)^3}{h(a,q)},\\
    f_{11}(y) = \frac{(2 a^2 + q)\left[(4  a^2 + 3 q)^2 - y  (8  a^4 + 12  a^2  q + 3  q^2)\right]}{2 h(a, q)}(1-y),\\
    f_{12}(y) = \frac{(2 a^2 + q)(8  a^4 + 12  a^2  q + 3  q^2)}{2 h(a,q)} y(1-y),\\
    f_{22}(y) =\frac{(2 a^2 + q)\left[(4  a^2 + 3 q)^2 - (1-y)  (8  a^4 + 12  a^2  q + 3  q^2)\right]}{2 h(a, q)}y, 
\end{gather}
\end{subequations}
where we have defined $h(a,q) = 4 a^4 + 6 a^2 q + 3 q^2$. Setting $q =0$ recovers the homogeneous vacuum solution. 

\subsection{Minimal supergravity}\label{app:knownsugra}
Within minimal supergravity, there is a charged version of the Myers-Perry black hole \cite{Chong:2005hr}. The extremal solutions come in two branches. Adapting the presentation in \cite{Horowitz:2024kcx} to our Ansatz, they are given by
\begin{subequations}
\begin{gather}
\psi(y) = - \frac{q \sqrt{3 (a b + q)\hat{\delta}}}{\Gamma(y)^2}, \quad B_1(y) =  \frac{a q ( a (a + b \hat{\delta}) + q\hat{\delta}) \sqrt{3}}{2\Gamma(y)^2}, \quad B_2(y) = -\frac{\hat{\delta} b q ( b (a + b \hat{\delta}) + q) \sqrt{3}}{2\Gamma(y)^2},\\
\Gamma(y) = (a^2 - b^2) y + b (a \hat{\delta} + b) + q \hat{\delta}, \\
\omega_1 = -\frac{2 \hat{\delta}\left(b(a+ \hat{\delta} b)^2+q(2b \hat{\delta}+a)\right)}{\left((a \hat{\delta} +b)^2+ q \hat{\delta}\right)\sqrt{(ab+q)\hat{\delta}}}, \quad \omega_2 = -\frac{2 \hat{\delta} \left(a(a \hat{\delta}+b)^2+q(2a \hat{\delta}+b)\right)}{\left((a+b \hat{\delta})^2+q \hat{\delta} \right)\sqrt{(ab+q)\hat{\delta}}},\\
f_{11}(y) = \frac{W(y)}{\Gamma(y)^2}(1-y), \quad f_{12}(y) = \frac{Y(y)}{\Gamma(y)^2}y(1-y), \quad f_{22}(y) = \frac{Z(y)}{\Gamma(y)^2}y.
\end{gather}
\end{subequations}
Here we defined
\begin{subequations}
\begin{align}
    W(y) &= (b-a \hat{\delta})(a+b\hat{\delta})^3(ab+q)y^2+\hat{\delta}(ab+q)((a+b\hat{\delta})^2+q\hat{\delta})^2 \\
    &\hspace{.5cm}+ ((a^2-2ab\hat{\delta}-2b^2)q^2 + \hat{\delta}(a+b\hat{\delta})^2(3a^2-4ab\hat{\delta}-2b^2)q+a(a^2-2b^2)(a+b\hat{\delta})^3)y,\nonumber\\
    Y(y) &= (a-b\hat{\delta})(a+ b\hat{\delta})^3(ab+q)y + \hat{\delta}(a^2+ab\hat{\delta}+b^2)q^2\\
    &\hspace{.5cm}+\hat{\delta}b(b\hat{\delta}+2a)(a+b\hat{\delta})^2q+\hat{\delta}ab^2(a+b\hat{\delta})^3,\\
    Z(y) &= \hat{\delta}(a-b\hat{\delta})(a+b\hat{\delta})^3(ab+q)y^2 + (b(a\hat{\delta}+b)+q\hat{\delta})^3\nonumber\\
    &\hspace{.5cm}+ ((2a^2+2ab\hat{\delta}-b^2)q^2+(4a-b\hat{\delta})b(a+b\hat{\delta})^2q+(2a-b\hat{\delta})b^2(a+b\hat{\delta})^3)y,
\end{align}
\end{subequations}
where $a, b, q$ are real constants and
\be
\hat{\delta} = 
\begin{cases}
    +1, & ab+q>0\\
    -1, & ab+q<0
\end{cases}.
\ee
The extremal Myers-Perry solution is obtained by setting $q = 0$. Setting $ a = b$, we recover two two-parameter subfamilies of homogeneous solutions. 

\bibliographystyle{ourbst}
\bibliography{references}

\end{document}